\documentclass[10pt,journal,letterpaper]{IEEEtran}
\ifCLASSINFOpdf
\else
  \usepackage[dvips]{graphicx}
\fi

\usepackage{url}
\usepackage{graphicx}
\graphicspath{{./}{figures/}}
\usepackage[caption=false,subrefformat=parens,labelformat=parens]{subfig}
\usepackage{cite}
\usepackage{amsmath,amssymb,amsfonts,amsthm}
\usepackage{bbm}
\usepackage{algorithm}
\usepackage{algpseudocode}
\usepackage{mathtools}
\usepackage{braket}
\usepackage{bm}
\usepackage{multirow}
\usepackage{booktabs}
\usepackage{xcolor}
\makeatletter
\newcommand{\algwrap}[2][]{\setbox0\hbox{#1}\parbox[t]{\dimexpr\linewidth-\ALG@thistlm-\wd0\relax}{#2\strut}}
\makeatother
\usepackage{flushend}

\newtheorem{lemma}{Lemma}

\begin{document}

\title{Robust Beamforming Design for Coherent Distributed ISAC with Statistical RCS and Phase Synchronization Uncertainty}

\author{Seonghoon Yoo,~\IEEEmembership{Graduate Student Member,~IEEE},
Seulhyun Kwon, Kawon Han,~\IEEEmembership{Member,~IEEE},
Elaheh Ataeebojd,~\IEEEmembership{Member,~IEEE},
Mehdi Rasti,~\IEEEmembership{Senior Member,~IEEE}, and Joonhyuk Kang%
\thanks{This work was supported by Institute of Information \& Communications Technology Planning \& Evaluation (IITP) grant funded by the Korea government (MSIT) (No. RS-2024-00444230, Development of Wireless Technology for Integrated Sensing and Communication).}
\thanks{This work has been submitted to the IEEE for possible publication. Copyright may be transferred without notice, after which this version may no longer be accessible.}
\thanks{Seonghoon Yoo, Seulhyun Kwon and Joonhyuk Kang are with the School of Electrical Engineering, Korea Advanced Institute of Science and Technology, Daejeon 34141, South Korea (e-mail: shyoo902@kaist.ac.kr, seulhyun@kaist.ac.kr, jhkang@ee.kaist.ac.kr).
Kawon Han is with the Department of Electrical Engineering, Ulsan National Institute of Science and Technology, Ulsan, South Korea (e-mail: kawon.han@unist.ac.kr).
E. Ataeebojd and M. Rasti are with the Centre for Wireless Communications, University of Oulu, Oulu, Finland (e-mail: \{elaheh.ataeebojd, mehdi.rasti\}@oulu.fi). \textit{(Corresponding author: Joonhyuk Kang.)}}}

\maketitle

\begin{abstract}
Distributed integrated sensing and communication (D-ISAC) enables multiple spatially distributed nodes to cooperatively perform sensing and communication. However, coherent cooperation across distributed nodes is challenging due to practical impairments. In particular, residual phase synchronization errors result in imperfect channel state information (CSI), while uncertainty in the angle of arrival (AoA) affects the assumed radar cross section (RCS) and array responses. To address these challenges, we propose a robust beamforming design for detecting a low-observable airborne target while supporting multi-user communication. The distributed nodes perform coherent joint transmission coordinated multipoint (JT-CoMP) communication and multistatic multiple-input multiple-output (MIMO) radar sensing. We maximize a lower bound on the expected sum of receive-node Kullback--Leibler divergence (KLD) terms for the worst-case AoAs on a joint grid, subject to power and per-user signal-to-interference-plus-noise ratio (SINR) constraints under imperfect CSI. The problem is solved using semidefinite relaxation (SDR) and successive convex approximation (SCA). Numerical results show that the proposed design achieves a higher worst-case detection probability than both zero-forcing beamforming and a non-robust scheme designed using nominal CSI and estimated AoAs.
\end{abstract}


\section{Introduction}
\IEEEPARstart{I}{ntegrated} sensing and communication (ISAC) uses shared radio resources to sense targets while serving user equipment (UE), improving spectral efficiency and hardware utilization \cite{Liu2022ISAC, Zhang2022ISAC}. Sharing transmit resources leads to a trade-off between sensing performance and communication quality of service (QoS) \cite{Han2025ISAC_tradeoff, Liu2018ISAC_tradeoff}. In most existing ISAC systems, communication and sensing are typically performed by a single base station (BS). Although such systems enable the joint operation of the two functions, their performance is inherently limited by the resources and spatial degrees of freedom available at a single node. In particular, sensing performance may degrade due to limited spatial diversity and coverage, while communication performance is constrained by the transmit power and antenna resources when serving multiple UEs \cite{Liu2018ISAC_tradeoff}. When ground-based BSs monitor a low-flying unmanned aerial vehicle (UAV), the same transmit power must support the downlink and illuminate the target from several viewing angles \cite{Li2025mission}. Distributed integrated sensing and communication (D-ISAC) addresses these limitations through cooperation among spatially separated nodes, providing spatial diversity for sensing and cooperative transmission gains for communication \cite{Babu2024Uncertainty_CSI, Cheng2024NetworkedISAC}.
\begin{table*}[t]
    \centering
    \caption{Comparison with Representative D-ISAC Studies}
    \label{related works}
    \begin{tabular}{lcccccl}
        \toprule
        \textbf{References}  & \textbf{D-ISAC} & 
        \textbf{Objective} & 
        \textbf{Synchronization error} &
        \textbf{Imperfect AoA} & \textbf{Statistical RCS} &  \\
        \midrule

        \cite{Babu2024Uncertainty_CSI}     & \checkmark          & Minimize CRB while maximizing SINR          & --          & --  & --               \\ 
        \cite{han2025distributed_ISAC}     & \checkmark          & Minimize CRB           &   --        & --  & --               \\
        \cite{Cheng2024NetworkedISAC}     & \checkmark         & Maximize detection probability            & Time          & --  & --               \\

        \cite{han2025over}     & \checkmark          &   Analyze synchronization error impact        & Time-frequency          & --  & --               \\
        \cite{Han2025phase_synchronization}     & \checkmark          & Analyze synchronization error impact          & Phase          & --  & --               \\
          \cite{Yang2024Networked_ISAC}     & \checkmark          & Maximize sum-rate         & Time          & --  & --               \\
         \cite{Xiu2026sync_imperfectAoA}     & \checkmark          & Minimize power consumption          & Time          & \checkmark  & --               \\
         \cite{Rivetti2025statistical_RCS}     & \checkmark         & Maximize sensing SCNR          & --          & --  & \checkmark               \\
\midrule
This paper & \checkmark & Maximize KLD & Phase (downlink) & \checkmark & \checkmark \\

        \bottomrule
    \end{tabular}
\end{table*}

In such distributed architectures, sensing can be implemented through monostatic or bistatic links among multiple network nodes \cite{Babu2024Uncertainty_CSI}. By exploiting the cooperative sensing capability of distributed links, D-ISAC can further enable multi-static sensing, which improves sensing reliability and coverage while supporting communication services for multiple UEs over a wide-area network. The authors in \cite{han2025distributed_ISAC} investigated orthogonal frequency-division multiplexing (OFDM) signaling for joint target localization and multi-user communication, while \cite{Cheng2024NetworkedISAC} studied coordinated transmit beamforming in networked ISAC systems to jointly optimize sensing and communication performance. Joint target assignment and beamforming for aerial vehicle detection were studied in \cite{Li2025mission}, while \cite{Singh2025Detection} considered target detection in a cell-free multiple-input multiple-output (MIMO) ISAC system. Sensing-assisted beam tracking in cell-free ISAC was investigated in \cite{Chen2026CellFree}. Moreover, \cite{Lou2025D_ISAC} investigated beamforming optimization in a D-ISAC system with integrated active and passive sensing, where sensing information collected by multiple nodes is fused at a central controller to improve target detection performance. In many of these systems, coherent cooperation among distributed nodes is typically assumed in order to fully exploit the spatial diversity and signal combining gains offered by multi-node sensing and transmission.

However, achieving coherent cooperation in practical D-ISAC systems is challenging due to various system impairments among distributed nodes. For example, synchronization offsets can be estimated using reference signals exchanged among distributed nodes \cite{han2025over}. Nevertheless, residual phase errors inevitably remain after the correction process and can degrade coherent signal combining across distributed nodes \cite{Han2025phase_synchronization, Dario2024perfect_synchronization}. Such impairments can be interpreted as channel state information (CSI) uncertainty. The authors in \cite{Han2025phase_synchronization} analyzed the impact of synchronization errors and phase noise on the sensing performance of D-ISAC systems, while \cite{Yang2024Networked_ISAC}  investigated coordinated beamforming in networked ISAC systems under imperfect CSI and time synchronization. Robust ISAC transmission under imperfect CSI was also studied for reconfigurable intelligent surface (RIS)-assisted systems \cite{Xu2024RobustRIS}. Transmitter time offsets have also been considered in multistatic target localization \cite{Wu2025MultistaticSynchronization}. Recent movable-antenna-aided ISAC studies have improved robustness against time synchronization errors and imperfect CSI via joint antenna-position and beamforming optimization \cite{Xiu2026Movable_timesync, Xiu2026sync_imperfectAoA}. However, the solution to mitigate the impact of phase synchronization errors in D-ISAC has been unexplored in coherent cooperation.

Another important factor affecting sensing performance in D-ISAC systems is the radar cross section (RCS) of the target. In practical scenarios, the RCS of UAVs can vary significantly with the observation angle due to their structural and material characteristics \cite{Ezuma2022UAV_RCS,khalili2025Angle_RCS}. In addition, stealth platforms may employ radar-absorbing materials (RAMs) to suppress radar reflections and reduce their detectability \cite{Ullah2024RCS_Reduction, Zheng2024EWAM}. Joint processing of monostatic and bistatic echoes has also been investigated for detecting low-observable targets \cite{Yan2026LowRCSDetection}. Due to such angle-dependent scattering effects, the received sensing signal strength can fluctuate significantly, and the target RCS is often modeled as a random variable that follows a statistical distribution \cite{Shi2012ChiLognormal}. This issue becomes more pronounced when the angle of arrival (AoA) of the target is imperfectly estimated, since AoA estimation errors can cause mismatches between the assumed and actual RCS characteristics. A summary of the key characteristics considered in existing D-ISAC studies is provided in Table~\ref{related works}.

Building on these observations, this paper proposes a robust beamforming design for coherent D-ISAC systems, where multiple distributed nodes are connected to a central unit (CU) and cooperatively perform joint transmission coordinated multipoint (JT-CoMP) communication and distributed MIMO radar sensing. In the considered scenario, the CU coordinates the distributed transceivers to detect a single target while simultaneously serving multiple UEs. Since the sensing task is formulated as centralized target detection, the Kullback--Leibler divergence (KLD) is adopted to measure the statistical separability between the target-absent and target-present hypotheses \cite{chen2025sensingKLD}. The sensing objective uses phase-invariant receive-node covariances, while the downlink signal-to-interference-plus-noise ratio (SINR) constraints account for phase-induced CSI uncertainty. Based on this model, we develop a robust D-ISAC beamforming framework that captures the sensing–communication trade-off under practical uncertainties.
The main contributions of this paper are summarized as follows:

\begin{itemize}
\item We consider a D-ISAC system in which distributed nodes coordinated by a CU detect a target with statistical RCS while serving multiple UEs through coherent JT-CoMP transmission. The nodes transmit orthogonal sensing waveforms, and the CU processes the multistatic echoes. The model accounts for phase-induced CSI errors on the downlink and for the effects of AoA uncertainty on the array responses and bistatic RCS.

\item We formulate a robust beamforming problem that maximizes a lower bound on the expected sum of receive-node KLDs for the worst-case AoAs on a joint grid. To guarantee the minimum SINR for each UE, we bound the desired signal power from below and the interference-plus-noise power from above over the CSI uncertainty sets. The resulting formulation captures the trade-off between target detection and communication QoS under a total power constraint. To the best of our knowledge, this is the first work that adopts the KLD-based sensing metric incorporating system impairments in D-ISAC systems for robust target detection.

\item To solve the resulting non-convex problem, we develop an efficient robust beamforming algorithm based on semidefinite relaxation (SDR) and successive convex approximation (SCA). The proposed approach jointly optimizes the communication and sensing beamforming vectors across distributed nodes while satisfying SINR requirements and system power constraints.

\item Numerical results demonstrate the convergence of the proposed algorithm and show how AoA uncertainty affects the sensing--communication trade-off. Under node phase errors, the proposed design achieves a higher detection probability than the non-robust and zero-forcing schemes at the same false-alarm probability. We also evaluate detection with the beamformers and receiver held fixed when the actual RCS profile differs from the one used for design.
\end{itemize}

\textit{Notation:} $\mathbb{E}[\cdot]$ denotes the expectation of the argument matrix; we reserve the superscript $A^T$ and $A^H$ for the transpose and the conjugate transpose of $A$, respectively; $\mathbf{0}_{N\times M}$ and $\mathbf{I}_N$ denote a $N\times M$ zero-matrix and a $N\times N$ identity matrix, respectively; $\mathbbm{1}\{\cdot\}$ denotes the indicator function, which equals one if the specified condition is satisfied and zero otherwise; and $\mathrm{tr}(\cdot)$, $\mathrm{diag}(\cdot)$, $\mathrm{blkdiag}(\cdot)$ and $\mathrm{rank}(\cdot)$ represent the trace, diagonal, block diagonal and rank function, respectively.

\section{System Model}
\label{system model}

\begin{figure}[t]
\centering
   \includegraphics[width=1\linewidth]{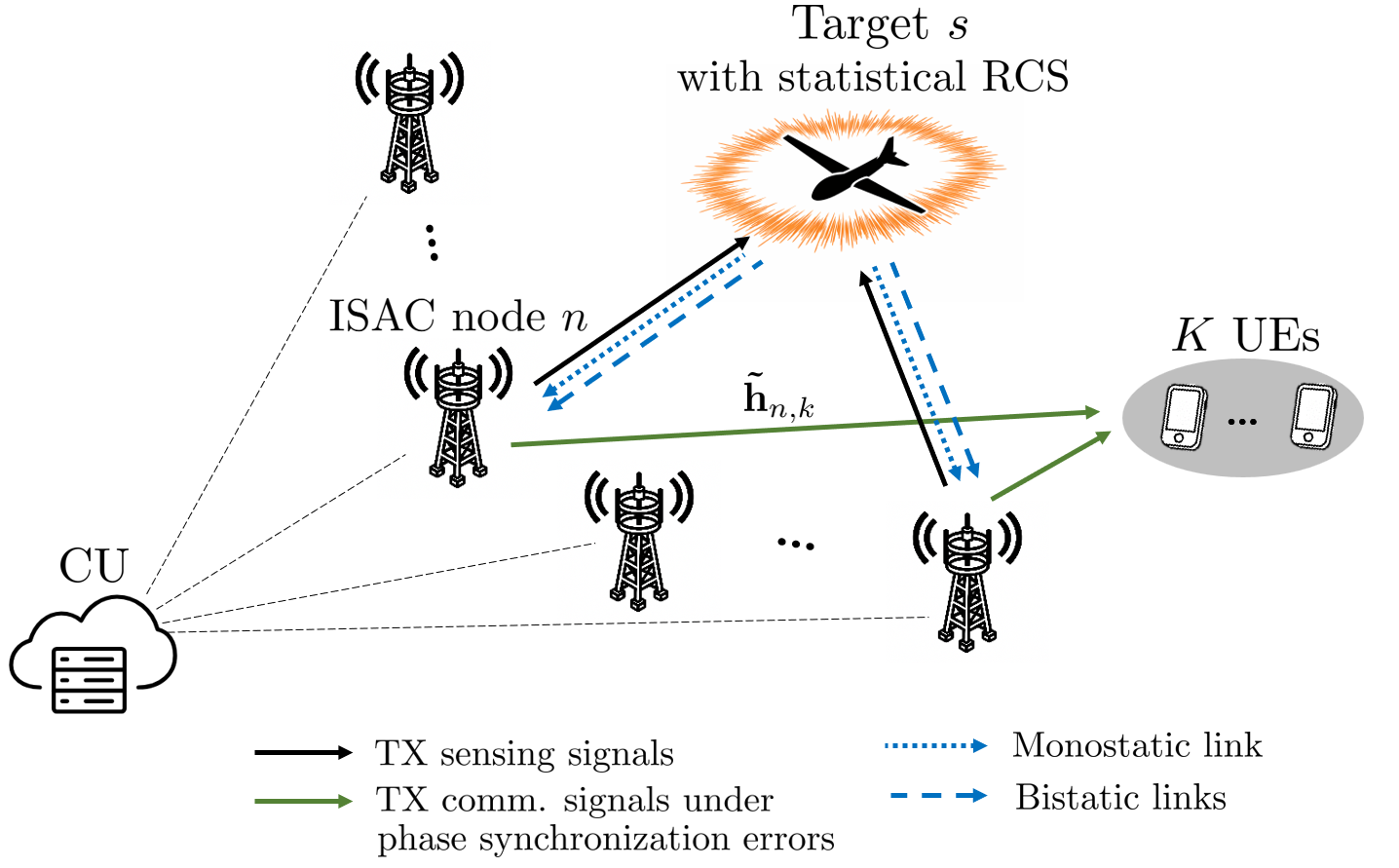}
   \caption{System model of the considered robust coherent D-ISAC network. Multiple spatially distributed ISAC nodes collaboratively transmit superposed communication and sensing waveforms and forward the received echoes to a CU for centralized processing.}
   \label{system_model}
\end{figure}

\subsection{Coherent D-ISAC System}

We consider a coherent distributed MIMO ISAC system, where $N$ spatially distributed ISAC nodes collaboratively perform joint sensing and communication, as illustrated in Fig.~\ref{system_model}. 
The distributed nodes exploit multi-static sensing links while supporting coherent JT-CoMP communication. Each node transmits an ISAC waveform that superposes communication and sensing components, while simultaneously receiving the reflected target echoes in a full-duplex manner\footnote{Residual self-interference and inter-node leakage after cancellation are included in the effective receiver noise following standard full-duplex modeling \cite{Sabharwal2014Full_duplex}.}. The received baseband signals at all ISAC nodes are forwarded to a CU equipped with a central processing unit (CPU) via fronthaul links. The CU uses the receive-node covariance blocks for beamforming design.

The $n$-th ISAC node is located at $\mathbf{p}_n=[x_n,y_n,0]^T$, for $n \in \{1,\ldots,N\}\triangleq\mathcal{N}$, where all nodes are deployed on the ground plane. 
In addition to sensing functionality, the distributed ISAC network simultaneously supports downlink communications for $K$ UEs located at $\mathbf{p}_k=[x_k,y_k,0]^T$, for $k \in \{1,\ldots,K\}\triangleq \mathcal{K}$, each equipped with a single antenna. Each node is equipped with $M_t$ transmit antennas and $M_r$ receive antennas, both configured as uniform linear arrays (ULAs) with antenna spacing $d_0$.

We aim to detect a target with low-observable RCS characteristics, such as a stealth aircraft or a large-scale UAV with stealth features \cite{Yan2026LowRCSDetection}. Specifically, the target is assumed to be coated with RAMs to suppress radar reflections and reduce the RCS \cite{Ullah2024RCS_Reduction}. The effect of target motion is reflected through aspect-angle variations across distributed bistatic links, which change the statistical RCS profile observed by the ISAC nodes. The location of the sensing target is denoted by $\mathbf{p}_s = [x_s, y_s, z]^T$, where $z$ represents a known altitude. The distance from node $n$ to the target is $\lVert\mathbf{p}_s-\mathbf{p}_n\rVert$. For analytical clarity, we focus on a single-target detection scenario, which enables us to isolate the impact of the proposed robust ISAC beamforming framework without additional coupling effects introduced by multiple targets \cite{Liu2020single_target}. Furthermore, such a framework can be readily extended to a multi-target scenario \cite{behdad2024multi}. For the considered target with statistical RCS, variations along the elevation dimension are relatively small over the sensing interval, and thus the target motion is mainly characterized on the horizontal plane \cite{Zheng2024EWAM}.

The transmitted signal of the $n$-th ISAC node at snapshot $t \in \{1,\ldots,T\}$ is expressed as
\begin{equation}
    \mathbf{x}_n[t]= \mathbf{W}_{c,n}\mathbf{s}_{c}[t]+ \mathbf{w}_{s,n}s_{s,n}[t],
\end{equation}
where $\mathbf{W}_{c,n} \in \mathbb{C}^{M_t \times K}$ and $\mathbf{w}_{s,n} \in \mathbb{C}^{M_t \times 1}$ denote the communication precoder and sensing beamformer, respectively. 
The signals $\mathbf{s}_{c}[t]=[s_{c,1}(t),\ldots,s_{c,K}(t)]^T \in \mathbb{C}^{K \times 1}$ and $s_{s,n}[t]$ represent the multi-user communication symbols for JT-CoMP transmission and the sensing waveform transmitted by the $n$-th ISAC node, respectively. 
Without loss of generality, we assume the uncorrelated conditions for the communication and sensing signals, satisfying $\mathbb{E}\big[\mathbf{s}_c[t](\mathbf{s}_c[t])^H\big]=\mathbf{I}_K$, $\mathbb{E}\big[\mathbf{s}_{c}[t]{s_{s,n}^*[t]}\big]=\mathbf{0}_{K \times 1}$, $\mathbb{E}\big[{s_{s,n}[t]s^*_{s,n}[t]}\big]=1$ and $\mathbb{E}\big[{s_{s,n}[t]s^*_{s,m}[t]}\big]=0$ for all $n\neq m$. 
For simplicity, we omit the symbol index $t$ in the following discussions.

\subsection{Communication System Model}
\label{Communication System Model}
The downlink channel $\mathbf{h}_{n,k}\in \mathbb{C}^{M_t\times 1}$ from the $n$-th ISAC node to the $k$-th UE is assumed to follow a Rician fading model \cite{Yoo2024Rician}, which is expressed as
\begin{equation}
\begin{aligned}
&\mathbf{h}_{n,k}
=\\& \sqrt{\ell_{n,k}}
\Bigg(
\sqrt{\frac{\gamma}{\gamma+1}}
e^{-j\frac{2\pi}{\lambda}\|\mathbf{p}_n-\mathbf{p}_k\|}
\mathbf{a}(\theta_{n,k})+
\sqrt{\frac{1}{\gamma+1}}
\tilde{\mathbf{a}}
\Bigg),
\end{aligned}
\end{equation}
where $\theta_{n,k}$ is the departure angle measured from the transmit array axis of node $n$, $\ell_{n,k}$ denotes the large-scale path-loss factor between the $n$-th ISAC node and UE $k$, and the term $e^{-j{2\pi\lVert \mathbf{p}_n-\mathbf{p}_k \rVert}/{\lambda}}$ represents the deterministic phase shift induced by the propagation.
Here, $\gamma$ is the Rician factor and $\lambda$ denotes the signal wavelength. 
The line-of-sight (LoS) component is modeled as the array steering vector $\mathbf{a}(\theta)=[1, e^{j{2\pi}d_0 \cos(\theta)/{\lambda}},\ldots,e^{j{2\pi}d_0(M_t-1)\cos(\theta)/{\lambda}}]^T \in \mathbb{C}^{M_t \times 1}$, and $d_0$ is the transmit antenna spacing.
The non-line-of-sight (NLoS) component $\tilde{\mathbf{a}}$ is modeled as a circularly symmetric complex Gaussian (CSCG) random vector with unit variance.

In practical distributed deployments, perfect phase synchronization among spatially separated ISAC nodes is difficult to achieve. 
In particular, independent local oscillators at different nodes introduce residual carrier frequency and phase offsets, resulting in imperfect alignment of reference time and frequency \cite{han2025distributed_ISAC}. 
Such synchronization impairments cause degradation in coherent transmission performance and may lead to errors in range and velocity estimation in sensing tasks. To explicitly account for phase synchronization imperfections, the residual phase error at the $n$-th node is denoted by $\phi_n$. 
The effective downlink channel from the $n$-th ISAC node to the $k$-th UE is modeled as
\begin{equation}
    \label{syncronization_error_chanel}
    \tilde{\mathbf{h}}_{n,k}=e^{j\phi_n}\mathbf{h}_{n,k} 
    = \mathbf{h}_{n,k}+\mathbf{e}_{n,k},
\end{equation}
where $\mathbf{e}_{n,k}=(e^{j\phi_n}-1)\mathbf{h}_{n,k}$ represents the phase-induced CSI mismatch. 
The residual phase error $\phi_n$ is common to all downlink channels from node $n$.
We assume that the phase-induced CSI mismatch is bounded within a spherical set
\begin{equation}
\label{sync_error_bound}
\mathcal{E}_{n,k}=\{\mathbf{e}_{n,k}: \lVert \mathbf{e}_{n,k} \rVert^2\leq \delta^2\},
\end{equation}
where $\delta$ is the common radius of the downlink CSI uncertainty sets.
The spherical sets conservatively bound the phase-induced perturbations when $\delta$ is chosen according to \eqref{phase_csi_radius} \cite{Babu2024Uncertainty_CSI,Masouros2015Interference}.

Under the above model, the received baseband signal at UE $k$ is expressed as
\begin{equation}
\begin{aligned}
   y_k 
   &= \sum_{n=1}^{N}\mathbf{\tilde{h}}_{n,k}^H\mathbf{x}_n + n_k \\
   &= \sum_{n=1}^{N}\big(
   \underbrace{\mathbf{\tilde{h}}_{n,k}^H\mathbf{w}_{c,n,k}s_{c,k}}_{\text{desired signal}}
   + 
   \underbrace{\sum_{j \in \mathcal{K}\backslash \{k\}}
   \mathbf{\tilde{h}}_{n,k}^H\mathbf{w}_{c,n,j}s_{c,j}}_{\text{multi-user interference}}
   \big)
   \\&+
   \underbrace{\sum_{n=1}^{N}
   \mathbf{\tilde{h}}_{n,k}^H\mathbf{w}_{s,n}s_{s,n}}_{\text{sensing interference}}
   + n_k,
\end{aligned}
\end{equation}
where $\mathbf{w}_{c,n,k} \in \mathbb{C}^{M_t\times 1}$ denotes the $k$-th column of $\mathbf{W}_{c,n}$, i.e., $\mathbf{W}_{c,n}=[\mathbf{w}_{c,n,1},\ldots,\mathbf{w}_{c,n,K}]$, and $n_k \sim \mathcal{CN}(0,\sigma_c^2)$ represents additive white Gaussian noise (AWGN), where $\sigma_c^2$ denotes the noise variance.

By exploiting coherent combining across distributed ISAC nodes, the downlink SINR at UE $k$ is defined as
\begin{equation}
\text{SINR}_{k}=
\frac{\left|\sum\limits_{n=1}^N 
\mathbf{\tilde{h}}^H_{n,k}\mathbf{w}_{c,n,k}\right|^2}
{\sum\limits_{j\ne k}
\left|\sum\limits_{n=1}^N 
\mathbf{\tilde{h}}^H_{n,k}\mathbf{w}_{c,n,j}\right|^2
+
\sum\limits_{n=1}^N
\left|\mathbf{\tilde{h}}_{n,k}^H\mathbf{w}_{s,n}\right|^2
+
\sigma_c^2}.
\end{equation}
In the above expression, the numerator represents the coherently combined desired communication signal power at UE $k$, while the denominator consists of the aggregated multi-user interference, sensing-induced interference, and the UE's receiver noise.

\subsection{Sensing System Model}
\label{Sensing System Model}

\subsubsection{Imperfect Target AoA and Statistical RCS Model}

\begin{figure}[t]
\centering
\subfloat[Chi-square model.]{\includegraphics[width=0.485\columnwidth]{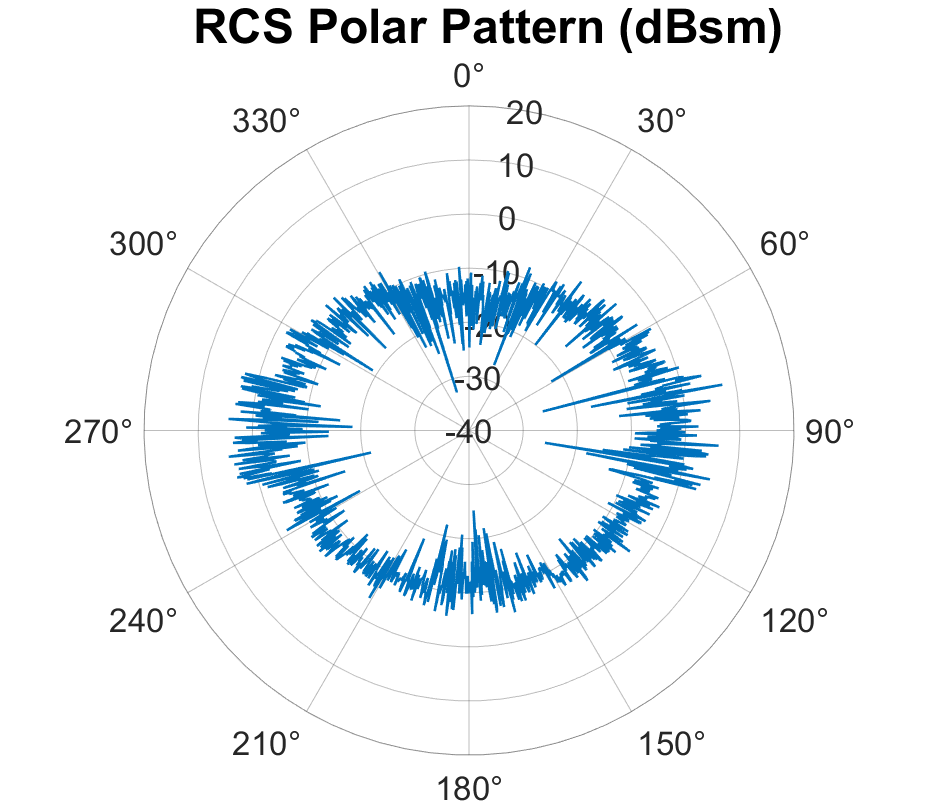}}\hfill
\subfloat[Swerling I model.]{\includegraphics[width=0.485\columnwidth]{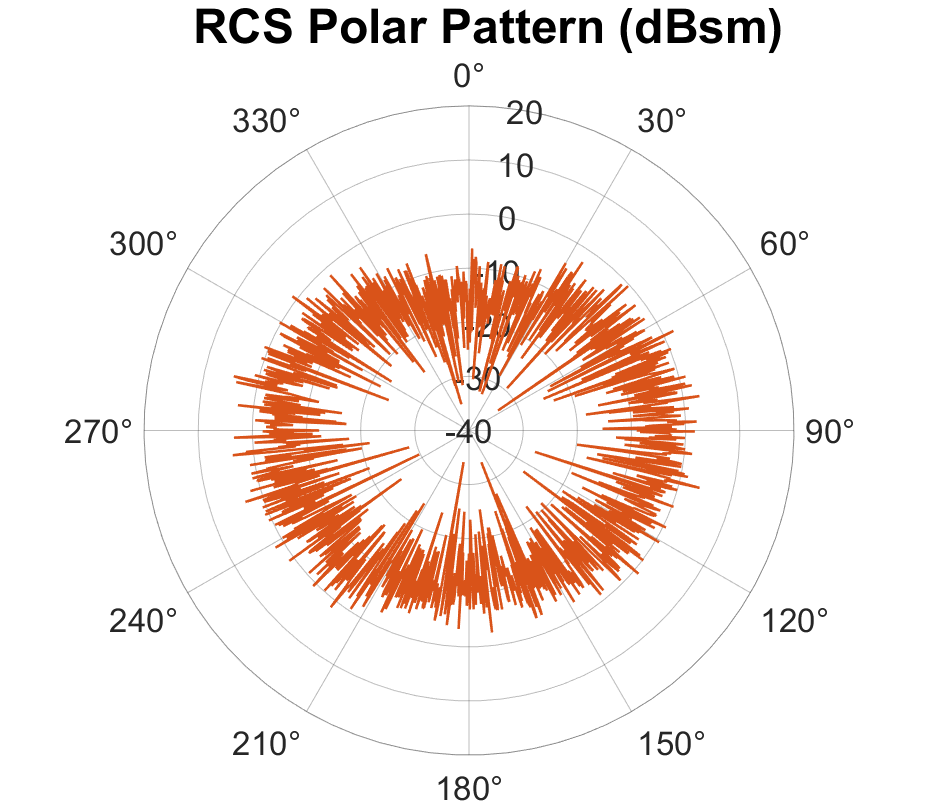}}

    \caption{Illustrative polar RCS patterns under different statistical RCS models for a target located at the center of the network: (a) Chi-square model and (b) Swerling I model.}
    \label{RCS_distribution}
\end{figure}

We consider a target whose statistical RCS is associated with low-observable scattering characteristics. Due to non-uniform electromagnetic absorption and geometrical surface features, the reflected energy observed by distributed ISAC nodes varies with the aspect angle \cite{khalili2025Angle_RCS}. As a result, the RCS exhibits angle-dependent statistical variations determined by the target scattering characteristics and motion-induced aspect-angle changes.

To consider practical sensing scenarios, the AoA information available at the transmitter side is assumed to have uncertainty. Let $\theta_{s,n}$ denote the AoA between the $n$-th ISAC node and the sensing target, with $\Theta_{s,n}$ representing the corresponding AoA uncertainty set. The AoA is assumed to lie within a bounded uncertainty region \cite{Lyu2024angle_uncertainty}
\begin{equation}
\label{aoa_bound}
\theta_{s,n} \in 
[\hat{\theta}_{s,n}-\Delta\theta_{s,n}, 
 \hat{\theta}_{s,n}+\Delta\theta_{s,n}]
\triangleq \Theta_{s,n},
\end{equation}
where $\hat{\theta}_{s,n}$ is the estimated AoA and $\Delta\theta_{s,n}$ denotes the maximum estimation error.

Under such angular uncertainty, it is necessary to model the statistical properties of the target RCS. 
In general, the RCS variation stems from two principal factors:  
(i) deterministic changes caused by aspect-angle variation, and  
(ii) probabilistic fluctuations induced by surface scattering mechanisms at a fixed aspect angle \cite{mahafza2005radar, Shi2012ChiLognormal}. 
For the maneuvering target, small variations in the aspect angle can result in significant RCS fluctuations due to the sensitivity of electromagnetic scattering to the observation geometry \cite{khalili2025Angle_RCS}.
Since the exact aspect angle cannot be perfectly determined a priori, it is therefore reasonable to characterize the RCS as a random variable following an appropriate statistical distribution. Representative statistical RCS models include the generalized Chi-square model \cite{Shi2012ChiLognormal} and the classical Swerling I model \cite{Swerling1960TIT}, which capture different fluctuation behaviors of target reflectivity. Fig.~\ref{RCS_distribution} illustrates example polar RCS patterns under different statistical assumptions when the target is centrally located in the ISAC system. The angular RCS patterns vary across the considered models. In particular, the Chi-square model exhibits relatively smooth fluctuations over angle, whereas the Swerling I model shows more pronounced variations \cite{Fang2020SwerlingChi-2}. These profiles provide the angle-dependent mean RCS powers for beamforming design. In Section~\ref{Numerical Simulation Results}-\ref{Impact of Power Constraint Setting and RCS Modeling}, we use RCS realizations from these profiles to evaluate the effect of RCS mismatch on detection.

Since electromagnetic scattering depends jointly on the incident and observation angles, the RCS between the $m$-th transmitting node and the $n$-th receiving node differs from the purely monostatic RCS. 
For small to moderate bistatic angles, and assuming the target surface is electrically large with dominant specular reflections, we model the bistatic RCS at a fixed carrier frequency as
\begin{equation}
    \beta_{s,n,m}
    =
    \beta_{s}
    \cos\!\left(
    \frac{|\theta_{s,n}-\theta_{s,m}|}{2}
    \right),
\end{equation}
where $\beta_{s}$ denotes the monostatic RCS power evaluated at the bistatic mid-angle \cite{Kell1965MonotoBi}. The angle difference in the cosine is wrapped to $[-\pi,\pi]$.

Accordingly, we characterize the bistatic RCS $\beta_{s,n,m}$ by its mean and variance, defined as 
$\mu_{s,n,m}=\mathbb{E}[\beta_{s,n,m}]$ and 
$\nu^2_{s,n,m}=\mathrm{Var}[\beta_{s,n,m}]$, respectively. 
The angular dependence of the RCS is encapsulated through these mean and variance terms, while other scattering components, e.g., delay-dependent effects, are assumed independent of the AoAs and are therefore not explicitly considered. 

\textbf{Remark 1:} To enable statistical modeling of the unknown instantaneous RCS, we assume that its angle-dependent statistical profile is available a priori. Specifically, the global RCS profile, characterized by the mean $\mu_{s,n,m}$ and variance $\nu^2_{s,n,m}$ of $\beta_{s,n,m}$ for all $n,m \in \mathcal{N}$, is defined in a target-centered coordinate system and can be obtained offline through electromagnetic simulations, measurement-based databases, or prior target libraries when the target class, shape, or material information is available  \cite{Stealth_profile}. During sequential sensing operation, these statistics can also be refined using accumulated target-echo measurements over consecutive frames. Given the AoA information $\theta_{s,n}$ and $\theta_{s,m}$, the effective observation angles at the distributed ISAC nodes are transformed into this target-centered coordinate system, from which the corresponding RCS statistics $\mu_{s,n,m}$ and $\nu_{s,n,m}^2$ are directly retrieved.

\subsubsection{Distributed Sensing Channel and Detection Model}
We first define the nominal sensing channel under perfect phase synchronization. Specifically, the effective sensing channel 
$\bar{\mathbf{G}}_{s,n,m} \in \mathbb{C}^{M_r \times M_t}$ 
from the $m$-th transmitting node to the $n$-th receiving node via reflection from a single-scatterer target\footnote{For an electrically large target with resolvable scattering centers, the model can be extended by summing multiple single-scatterer responses \cite{behdad2024multi}; here, we adopt the single-scatterer model for tractability.} is modeled as \cite{Demirhan2025ISAC_BF, behdad2024multi, Liu2021Sensing_channel}
\begin{equation}
\label{effective_sensing_channel}
\bar{\mathbf{G}}_{s,n,m}
=
L_{n,m}\sqrt{\beta_{s,n,m}}\,
\mathbf{A}_{n,m}(\theta_{s,n},\theta_{s,m}),
\end{equation}
where $L_{n,m}$ denotes the large-scale propagation factor defined as 
$
L_{n,m}
\triangleq
\sqrt{P_{n,m}}\,
e^{-j\frac{2\pi}{\lambda}
\left(
\lVert \mathbf{p}_n-\mathbf{p}_s \rVert
+
\lVert \mathbf{p}_s - \mathbf{p}_m\rVert
\right)}$,
with $P_{n,m}$ representing the large-scale path-loss coefficient of the bistatic link from node $m$ to the target and from the target to node $n$. The exponential term captures the deterministic phase rotation induced by the round-trip propagation distance. The angular response matrix is defined as
$
\mathbf{A}_{n,m}(\theta_{s,n},\theta_{s,m})
\triangleq
\mathbf{a}_r(\theta_{s,n})
\mathbf{a}_t^{T}(\theta_{s,m}),
$
which characterizes the receive and transmit array responses toward the target. The sensing steering vectors use the projection of the line of sight onto the horizontal array axis, with azimuth uncertainty and elevation fixed by the target height.

In practice, residual phase synchronization errors exist among distributed ISAC nodes. 
As a result, the effective sensing channel experiences an additional relative phase rotation and can be written as
\begin{equation}
\mathbf{G}_{s,n,m}=e^{j(\phi_n-\phi_m)}\bar{\mathbf{G}}_{s,n,m}.
\end{equation}
The same node phases enter the downlink channel in \eqref{syncronization_error_chanel}. Since $\|\mathbf e_{n,k}\|=2|\sin(\phi_n/2)|\,\|\mathbf h_{n,k}\|$, the spherical sets in \eqref{sync_error_bound} contain all perturbations with $|\phi_n|\leq\phi_{\max}\leq\pi$ when
\begin{equation}
\label{phase_csi_radius}
\delta\geq 2\sin(\phi_{\max}/2)\max_{n,k}\|\mathbf h_{n,k}\|.
\end{equation}
The phase rotation cancels when forming the covariance of the echo from transmit node $m$, yielding
\begin{equation}
\mathbf G_{s,n,m}\mathbf w_{s,m}\mathbf w_{s,m}^H\mathbf G_{s,n,m}^H
=\bar{\mathbf G}_{s,n,m}\mathbf w_{s,m}\mathbf w_{s,m}^H\bar{\mathbf G}_{s,n,m}^H.
\end{equation}

We assume that each ISAC node employs mutually orthogonal sensing waveforms, whose waveform information is shared via the CU. The snapshot model below retains the echoes from the dedicated sensing waveforms. Their second-order orthogonality removes terms between different transmit waveforms from the covariance. Residual communication echoes and leakage after suppression are included in the effective clutter-plus-noise term \cite{Zhang2021Matched_filter}, which is modeled with a fixed covariance during beamforming design.

Under the above model, the target detection problem at the CU, which combines all received sensing signals $y_{s,n}$ for all ${n\in \mathcal{N}}$, can be formulated as a binary hypothesis test \cite{Cheng2024NetworkedISAC}.
Depending on the presence or absence of the target, it is given by
\begin{subequations}
\begin{align}
    \mathcal{H}_0 &:~
    \mathbf{y}_{s,n}
    = \sum_{m=1}^{N}\mathbf{c}_{n,m}
    + \mathbf{n}_{s,n}, \label{eq:hypothesis_H0} \\
    \mathcal{H}_1 &:~
    \mathbf{y}_{s,n}
    = \sum_{m=1}^{N}
    \big(
    \mathbf{G}_{s,n,m}
    \mathbf{w}_{s,m}
    s_{s,m}
    + \mathbf{c}_{n,m}
    \big)
    + \mathbf{n}_{s,n}, \label{eq:hypothesis_H1}
\end{align}
\end{subequations}
where $\mathcal{H}_0$ and $\mathcal{H}_1$ correspond to the hypotheses of target absence and presence, respectively. The term $\mathbf{c}_{n,m} \in \mathbb{C}^{M_r \times 1}$ denotes the clutter component, accounting for undesired reflections from surrounding objects. 
By invoking the central limit theorem, the aggregate clutter is modeled as a complex Gaussian random vector
$\mathbf{c}_{n,m} \sim \mathcal{CN}(\mathbf{0}, \mathbf{R}^{\text{clu}}_{n,m})$ \cite{Jeong2025PMN}, where $\mathbf{R}^{\mathrm{clu}}_{n,m}$ denotes the clutter covariance matrix corresponding to the signal transmitted from the $m$-th ISAC node and received at the $n$-th sensing node. The noise vector $\mathbf{n}_{s,n} \in \mathbb{C}^{M_r\times 1}$ represents circularly symmetric complex Gaussian noise following 
$\mathcal{CN}(\mathbf{0}, \sigma_{s,n}^2\mathbf{I}_{M_r})$. 
For analytical tractability, we assume identical noise variance across nodes, i.e., $\sigma_{s,n}^2 = \sigma_s^2$.

\section{KLD Derivation for Centralized Target Detection Framework}
\label{KLD Derivation for Centralized Target Detection Framework}

In the proposed D-ISAC architecture, the CU evaluates a centralized detection criterion using the sensing signals received from all ISAC nodes. The aggregated sensing signal collected at the CU is defined as
\begin{equation}
\mathbf{y}_s =[
\mathbf{y}_{s,1}^T,
\ldots,
\mathbf{y}_{s,N}^T
]^T
\in
\mathbb{C}^{NM_r \times 1}.
\end{equation}
To quantify the target detection performance, we adopt the KLD between the conditional distributions under the two hypotheses as the sensing performance metric. Unlike metrics based on the target echo power relative to clutter and noise or detection probability that depends on the decision threshold, KLD captures the distributional separability of the aggregated sensing signal and provides a threshold-independent measure of detection performance \cite{chen2025sensingKLD}. The aggregated received signal can be rewritten as
\begin{equation}
\mathbf{y}_s  =
\begin{cases}
\mathbf{c}+\mathbf{n}, & \text{under } \mathcal{H}_0,\\[4pt]
\mathbf{r}_s+\mathbf{c}+\mathbf{n}, & \text{under } \mathcal{H}_1,
\end{cases}
\end{equation}
where $\mathbf{r}_s$ denotes the aggregated target echo,
\begin{equation}
\mathbf{r}_s
=
\begin{bmatrix}
\sum_{m=1}^{N}\mathbf{G}_{s,1,m}\mathbf{w}_{s,m}s_{s,m} \\
\vdots \\
\sum_{m=1}^{N}\mathbf{G}_{s,N,m}\mathbf{w}_{s,m}s_{s,m}
\end{bmatrix}
\in \mathbb{C}^{NM_r \times 1},
\end{equation}
and the aggregated clutter and noise vectors are defined as
\begin{equation}
\mathbf{c}
=
\big[
\sum_{m}\mathbf{c}_{1,m}^T,
\ldots,
\sum_{m}\mathbf{c}_{N,m}^T
\big]^T,
\end{equation}
\begin{equation}
\mathbf{n}
=
\big[
\mathbf{n}_{s,1}^T,
\ldots,
\mathbf{n}_{s,N}^T
\big]^T.
\end{equation}
Given the AoAs and RCS realization, we approximate each receive-node distribution by a complex Gaussian with the same mean and covariance. Their product gives the conditional model used for design, expressed as
\begin{equation}
\label{aggregated_received_signal_gaussian}
\mathbf{y}_s \sim
\begin{cases}
\mathcal{CN}(\mathbf{0}, \mathbf{R}_0), & \text{under } \mathcal{H}_0,\\[4pt]
\mathcal{CN}(\boldsymbol{\mu}_1, \mathbf{R}_1), & \text{under } \mathcal{H}_1.
\end{cases}
\end{equation}
Under $\mathcal{H}_0$, the effective clutter and noise are modeled as independent across receive nodes, yielding
\begin{equation}
\label{R_0_eq}
\mathbf{R}_0
=
\mathrm{blkdiag}(
\mathbf{R}_{0,1},\ldots,\mathbf{R}_{0,N}),
\end{equation}
where
\begin{equation}
\label{R_0_n_eq}
\mathbf{R}_{0,n}
=
\sum_{m=1}^{N}\mathbf{R}^{\mathrm{clu}}_{n,m}
+
\sigma_s^2\mathbf{I}_{M_r}.
\end{equation}
Under $\mathcal{H}_1$, the mean vector is
\begin{equation}
\boldsymbol{\mu}_1
=
\mathbb{E}[\mathbf{r}_s]
=
\mathbf{0},
\end{equation}
due to the zero-mean signaling assumption and the statistical independence between the transmitted signals and the sensing channels. Consequently, the covariance matrix under $\mathcal{H}_1$ becomes
\begin{equation}
\label{covariance_total}
\mathbf{R}_1
=
\mathbf{R}_s
+
\mathbf{R}_0,
\end{equation}
where the product model retains the diagonal receive-node blocks,
\begin{equation}
\mathbf{R}_s
=
\mathrm{blkdiag}
(
\mathbf{R}_{s,1},\ldots,\mathbf{R}_{s,N}
),
\end{equation}
with
\begin{equation}
\label{covariance_sensing}
\mathbf{R}_{s,n}
=
\sum_{m=1}^N
\mathbf{G}_{s,n,m}
\mathbf{w}_{s,m}
\mathbf{w}_{s,m}^H
\mathbf{G}_{s,n,m}^H.
\end{equation}
The cross-covariance between receive nodes $n$ and $n'$ retains the phase factor $e^{j(\phi_n-\phi_{n'})}$. The design uses the diagonal covariance blocks, which are invariant to those phases. Since the observations are independent under $\mathcal H_0$, the KLD chain rule yields the sum of receive-node KLDs as a lower bound on the joint Gaussian KLD.

Under the product model, the KLD between $\mathcal H_1$ and $\mathcal H_0$ is given by
\begin{equation}
\label{KL_def}
D_{\text{KL}}(\mathcal{H}_1\lVert \mathcal{H}_0)
=
\mathbb{E}_{\mathcal{H}_1}
\!\left[
\log
\frac{p(\mathbf{y}_s|\mathcal{H}_1)}
{p(\mathbf{y}_s|\mathcal{H}_0)}
\right].
\end{equation}
For the complex Gaussian models considered here, the KLD between the two hypotheses reduces to
\begin{equation}
D_{\text{KL}}
=
\log
\frac{|\mathbf{R}_0|}
{|\mathbf{R}_1|}
+
\mathrm{tr}
\left(
\mathbf{R}_0^{-1}\mathbf{R}_1
\right)
-
NM_r.
\end{equation}
For these zero-mean Gaussian models, the KLD measures how the target echo covariance changes the distribution of the received signal relative to that under $\mathcal{H}_0$.

\textbf{Remark 2: }
The marginal KLD sum depends on the sensing beamformers, AoAs, and RCS. For a given set of beams, it is invariant to the node phases. Phase uncertainty affects beamforming through the downlink SINR constraints and affects coherent detection through the receive-node cross-covariances. As a result, the KLD becomes an uncertainty-aware functional of the sensing beamforming strategy. Maximizing the KLD therefore enhances the statistical distinguishability between $\mathcal{H}_1$ and $\mathcal{H}_0$ under practical impairments, leading to improved centralized detection performance.

\section{Proposed Robust D-ISAC Beamforming}
\label{Problem Formulation and Proposed Robust Beamforming Algorithm}

\begin{figure*}[!t]
\begin{align}
\min_{\mathbf{e}_{n,k}}
\bigg|
\sum_{n=1}^N
\tilde{\mathbf{h}}_{n,k}^H
\mathbf{w}_{c,n,k}
\bigg|^2
&\geq \mathrm{tr}(\mathbf Q_k\mathbf W_c^{(k)})-2\sqrt{N}\delta\,\|\mathbf W_c^{(k)}\mathbf h_k\|_2, \tag{30}\label{signal_min}
\end{align}
\begin{align}
\max_{\mathbf{e}_{n,k}}
&\bigg(
\sum_{j \in \mathcal{K}\backslash \{k\}}
\bigg|
\sum_{n=1}^N
\tilde{\mathbf{h}}_{n,k}^H
\mathbf{w}_{c,n,j}
\bigg|^2
+
\sum_{n=1}^N
\big|
\tilde{\mathbf{h}}_{n,k}^H
\mathbf{w}_{s,n}
\big|^2
+ \sigma_c^2
\bigg)
\nonumber\\
&\leq
\sum_{j \in \mathcal{K}\backslash \{k\}}
\bigg(
\mathrm{tr}\!\left(
\mathbf Q_k
\mathbf W_{c}^{(j)}
\right)
+ N\delta^2\,\mathrm{tr}(\mathbf W_{c}^{(j)})
+ \sqrt{N}\delta
\|\mathbf{W}_{c}^{(j)}\mathbf{h}_k\|
+ \sqrt{N}\delta
\|\mathbf{h}_k^H\mathbf{W}_{c}^{(j)}\|
\bigg)
\nonumber\\
&\quad
+
\sum_{n=1}^N
\bigg(
\mathrm{tr}(
\mathbf Q_{n,k}
\mathbf W_{s,n}
)
+ \delta^2\,\mathrm{tr}(\mathbf W_{s,n})
+ \delta
\|\mathbf{W}_{s,n}\mathbf{h}_{n,k}\|
+ \delta
\|\mathbf{h}_{n,k}^H\mathbf{W}_{s,n}\|
\bigg)
+ \sigma_c^2 , \tag{31}\label{interference_max}
\end{align}
\hrulefill
\vspace{-0.5cm}
\end{figure*}

\subsection{Problem Formulation}
\label{Problem Formulation}

Based on the communication and sensing performance metrics established in Sections~\ref{system model} and~\ref{KLD Derivation for Centralized Target Detection Framework}, we aim to design the transmit beamformers to maximize the KLD averaged over the statistical RCS coefficients while satisfying the downlink communication QoS requirements for all UEs. Assuming centralized coordination at the CU, the distributed ISAC nodes jointly optimize their beamformers under a total transmit power constraint. The resulting problem is given by
\begin{subequations}\label{p1}
\begin{align}
\text{(P1):} \quad
& \max_{\{\mathbf{W}_{c,n}, \mathbf{w}_{s,n}\}}
\;\mathbb{E}_{\boldsymbol{\beta}}\!\left[D_{\text{KL}}\right] \label{p1a} \\[4pt]
\text{s.t.} \quad
& \text{SINR}_k \ge \Gamma_c, \quad \forall k, \label{p1b} \\[4pt]
& \sum_{n=1}^N 
\mathrm{tr}\!\left(
\mathbf{W}_{c,n}\mathbf{W}_{c,n}^H
+
\mathbf{w}_{s,n}\mathbf{w}_{s,n}^H
\right)
\le P^{\max}, \label{p1c} \\[4pt]
& \eqref{sync_error_bound},\; \eqref{aoa_bound}, \label{p1d}
\end{align}
\end{subequations}
where $\boldsymbol{\beta}=\{\beta_{s,n,m}\}_{n,m \in \{1,\ldots,N\}}$ denotes the collection of random RCS coefficients and $P^{\max}$ denotes the total transmit power budget across all distributed ISAC nodes.

For a given AoA tuple, the objective in \eqref{p1a} averages the KLD over RCS realizations. The following subsection derives a lower bound and formulates its worst-case optimization over a joint AoA grid.
The constraint \eqref{p1b} specifies the downlink SINR requirement $\Gamma_c$ for each UE. This constraint must hold for every collection of channel errors $\{\mathbf e_{n,k}\}_{n=1}^{N}$ with $\mathbf e_{n,k}\in\mathcal E_{n,k}$. The constraint \eqref{p1c} limits the total transmit power across all distributed ISAC nodes. The constraints \eqref{p1d} account for the uncertainty caused by phase synchronization errors and imperfect target AoA knowledge.
Problem (P1) is highly non-convex due to the log-determinant objective in \eqref{p1a} and the fractional quadratic SINR constraints in \eqref{p1b}. 
To tackle this non-convexity, we reformulate and relax (P1) using SDR and SCA methods.

\subsection{Proposed Robust D-ISAC Beamforming}
\label{Proposed Robust D-ISAC Beamforming}

To account for the underlying system uncertainties in (P1), we develop a robust beamforming algorithm by sequentially handling phase synchronization errors, incorporating angular uncertainty and statistical RCS variations, and transforming the resulting non-convex problem into a tractable form via relaxation and convex approximation techniques.

\subsubsection{Phase Synchronization Error in Communication}
We first consider the phase synchronization errors in the downlink communication channels. 
Since the CSI perturbation $\mathbf{e}_{n,k}$ lies within the bounded uncertainty set $\mathcal{E}_{n,k}$ in (\ref{syncronization_error_chanel}) and (\ref{sync_error_bound}), the SINR constraint in (P1) must hold for all admissible channel realizations. This leads to the following conservative min--max SINR expression, where the desired-signal and interference-plus-noise terms are separately handled in the subsequent approximation, given by
\begin{equation}
\begin{aligned}
\label{min_max_SINR}
&\overline{\text{SINR}}_{k}=\\&
\frac{\min\limits_{\mathbf{e}_{n,k}}
\big\lvert{\sum\limits_{n=1}^N 
\mathbf{\tilde{h}}^H_{n,k}
\mathbf{w}_{c,n,k}}\big\rvert^2}
{\max\limits_{\mathbf{e}_{n,k}}
\bigg(
\sum\limits_{j\ne k}
\big\lvert \sum\limits_{n=1}^N 
\mathbf{\tilde{h}}^H_{n,k}
\mathbf{w}_{c,n,j}\big \rvert^2
+
\sum\limits_{n=1}^N
\big\lvert  
\mathbf{\tilde{h}}_{n,k}^H
\mathbf{w}_{s,n} \big\rvert^2
+
\sigma_c^2
\bigg)} .
\end{aligned} \tag{29}
\end{equation}
\setcounter{equation}{31}

To obtain a tractable approximation of \eqref{min_max_SINR}, we exploit its quadratic structure and follow the robust precoding approach in \cite{Babu2024Uncertainty_CSI}. 
Specifically, we use a lower bound on the desired-signal power and an upper bound on interference plus noise, as summarized in the following lemma.

\begin{lemma}
For a given CSI uncertainty set 
$\mathcal{E}_{n,k}=\{\mathbf{e}_{n,k}:\lVert \mathbf{e}_{n,k} \rVert^2\leq \delta^2\}$, 
the right-hand side of \eqref{signal_min} lower-bounds the desired power, with
$\mathbf Q_k \triangleq \mathbf h_k \mathbf h_k^H$. 
The interference-plus-noise term satisfies the upper-bound representation in \eqref{interference_max}, where 
$\mathbf Q_{n,k}\triangleq \mathbf h_{n,k}\mathbf h_{n,k}^H$. 
Here, $\mathbf h_k=[\mathbf h_{1,k}^{T},\ldots,\mathbf h_{N,k}^{T}]^T$, and the stacked effective channel vector and communication beamformer are defined as
$
\tilde{\mathbf{h}}_k
=
[\tilde{\mathbf{h}}_{1,k}^T,\ldots,\tilde{\mathbf{h}}_{N,k}^T]^T
$
and
$
\mathbf{w}_{c}^{(k)}
=
[\mathbf{w}_{c,1,k}^T,\ldots,\mathbf{w}_{c,N,k}^T]^T,
$
respectively. 
The lifted matrices are defined as
$
\mathbf{W}_{c}^{(k)}
=
\mathbf{w}_{c}^{(k)}
(\mathbf{w}_{c}^{(k)})^H
\in \mathbb{C}^{NM_t \times NM_t}
$
and
$
\mathbf{W}_{s,n}
=
\mathbf{w}_{s,n}
\mathbf{w}_{s,n}^H,
$ which are rank-one positive semidefinite matrices.
\end{lemma}
\begin{proof}
The proof follows from quadratic expansion under bounded CSI uncertainty and the Cauchy--Schwarz inequality. 
Both bounds are derived in Appendix~\ref{proof_lemma1}.
\renewcommand{\qedsymbol}{}
\end{proof}

To further facilitate the convex reformulation, we introduce the lifted variables
$\mathbf{W}_{c}^{(k)}
=
\mathbf{w}_{c}^{(k)}(\mathbf{w}_{c}^{(k)})^H$
and
$\mathbf{W}_{s,n}
=
\mathbf{w}_{s,n}\mathbf{w}_{s,n}^H
\in \mathbb{C}^{M_t \times M_t}$,
which are rank-one positive semidefinite matrices.
Let $I_k$ denote the right-hand side of \eqref{interference_max}, which upper-bounds the total interference-plus-noise power over the spherical CSI sets. Requiring the desired-signal lower bound in \eqref{signal_min} to be at least $\Gamma_c I_k$ ensures the SINR requirement for every CSI error in \eqref{sync_error_bound}. Using the lifted variables, we obtain the following conservative formulation of (P1)
\begin{subequations}\label{p2}
\begin{align}
\text{(P2):} \quad
& \max_{\{\mathbf{W}_{c}^{(k)}, \mathbf{W}_{s,n}\}}
\;\mathbb{E}_{\boldsymbol{\beta}}[D_{\text{KL}}] \\
\text{s.t.}\quad
& \mathrm{tr}(\mathbf Q_k\mathbf W_c^{(k)})-2\sqrt{N}\delta\|\mathbf W_c^{(k)}\mathbf h_k\|_2 \geq \Gamma_c I_k,\ \forall k,\\
& \sum_{k=1}^K\mathrm{tr}(\mathbf{W}_{c}^{(k)})
+\sum_{n=1}^N\mathrm{tr}(\mathbf{W}_{s,n})
\le P^{\max}, \\
& \mathbf{W}_{c}^{(k)}\succeq 0,\;
\mathrm{rank}(\mathbf{W}_{c}^{(k)})=1, \\
& \mathbf{W}_{s,n}\succeq 0,\;
\mathrm{rank}(\mathbf{W}_{s,n})\leq1, \\
& \eqref{aoa_bound}.
\end{align} 
\end{subequations}

The rank constraints preserve the representation in terms of the original beamforming vectors.
However, problem (P2) remains non-convex due to the non-linear objective function and the rank constraints.

\subsubsection{KLD Characterization under Statistical RCS Uncertainty}

Before solving (P2), the uncertainty arising from imperfect AoA information and stochastic RCS coefficients must be properly incorporated into the objective function (\ref{p2}a). 
Since the RCS coefficients $\boldsymbol{\beta}$ are random variables, the detection performance is characterized in terms of the average KLD with respect to their distribution.

The expectation of the KLD in (\ref{p2}a) can be expressed as
\begin{equation}
    \begin{aligned}
        \label{average_KLD}
        &\mathbb{E}_{\boldsymbol{\beta}}[D_{\text{KL}}]
        \\&=
        \log |\mathbf{R}_0|
        -
        \mathbb{E}_{\boldsymbol{\beta}}
        \!\left[
        \log |\mathbf{R}_1|
        \right]
        +
        \mathrm{tr}
        \!\left(
        \mathbf{R}_0^{-1}
        \mathbb{E}_{\boldsymbol{\beta}}[\mathbf{R}_1]
        \right)
        - NM_r
        \\
        &\ge
        \log |\mathbf{R}_0|
        -
        \log
        \left|
        \mathbb{E}_{\boldsymbol{\beta}}[\mathbf{R}_1]
        \right|
        +
        \mathrm{tr}
        \!\left(
        \mathbf{R}_0^{-1}
        \mathbb{E}_{\boldsymbol{\beta}}[\mathbf{R}_1]
        \right)
        - NM_r,
    \end{aligned}
\end{equation}
where the inequality in \eqref{average_KLD} follows from Jensen's inequality applied to the concave log-determinant term; therefore, we maximize the right-hand side of \eqref{average_KLD} as a tractable lower-bound objective, without treating it as an equivalent reformulation of the exact expected KLD.

Since $\mathbf{R}_0$ is deterministic, the randomness of the RCS coefficients only affects the target-related covariance. Hence,
\begin{equation}
\mathbb{E}_{\boldsymbol{\beta}}[\mathbf{R}_1]
=
\mathbf{R}_0
+
\mathbb{E}_{\boldsymbol{\beta}}[\mathbf{R}_s].
\end{equation}
Recalling \eqref{covariance_total}–\eqref{covariance_sensing}, the target-related covariance has a block-diagonal structure. The $n$-th diagonal block of $\mathbb{E}_{\boldsymbol{\beta}}[\mathbf{R}_s]$ is given by
\begin{equation}
\label{expected_sensing_covariance}
\begin{aligned}
\mathbb{E}_{\boldsymbol{\beta}}[\mathbf{R}_{s,n}]
&=
\sum_{m=1}^{N}
|L_{n,m}|^2
\mu_{s,n,m}
\\
&\quad
\mathbf{A}_{n,m}
\mathbf{W}_{s,m}
\mathbf{A}_{n,m}^H .
\end{aligned}
\end{equation}
RCS profiles with the same angle-dependent mean powers give the same Jensen objective. Their variances affect the KLD and detection statistics across realizations.

\subsubsection{Imperfect Target AoA Information}
Imperfect target AoA information affects both the array responses and the mean RCS powers in \eqref{expected_sensing_covariance}. Since each bistatic link depends on the transmit and receive angles, we consider the angles jointly. Denoting the lower bound in \eqref{average_KLD} by $F$, we maximize its worst-case value over the joint angle grid, given by
\begin{equation}
\label{worst_case_aoa}
\min_{(\theta_1,\ldots,\theta_N)\in\mathcal G}F(\theta_1,\ldots,\theta_N).
\end{equation}
The grid $\mathcal G$ contains the nominal angle combination and all sign combinations at selected common offsets, including the uncertainty bound. The minimum is taken over these sampled angle combinations. The initial beam covariances are obtained by maximizing total sensing power under the rank-relaxed QoS and power constraints.

\subsubsection{SCA-Based KLD Relaxation}

We next examine the curvature of the objective function in problem (P2).
From \eqref{average_KLD}, the lower bound consists of a constant term $\log|\mathbf{R}_0|$,
a linear trace term $\mathrm{tr}\big(\mathbf{R}_0^{-1}\mathbb{E}_{\boldsymbol{\beta}}[\mathbf{R}_1]\big)$, 
and a negative log-determinant term 
$-\log\big|\mathbb{E}_{\boldsymbol{\beta}}[\mathbf{R}_1]\big|$. 
Since $\mathbb{E}_{\boldsymbol{\beta}}[\mathbf{R}_1]$ is affine with respect to the beamforming matrices 
$\{\mathbf{W}_{c}^{(k)},\mathbf{W}_{s,n}\}$, 
the trace term is linear in the optimization variables. 
Furthermore, the function $-\log\big|\mathbb{E}_{\boldsymbol{\beta}}[\mathbf{R}_1]\big|$ 
is convex with respect to 
$\mathbb{E}_{\boldsymbol{\beta}}[\mathbf{R}_1]$, 
and consequently convex in the beamforming matrices. As a result, the overall objective function of (P2) is non-concave with respect to the optimization variables.

To construct the affine bound, we denote the expected covariance by
\begin{equation}
\mathbf{Z}\triangleq \mathbb{E}_{\boldsymbol{\beta}}[\mathbf{R}_1]
=\mathbf{R}_0+\mathbb{E}_{\boldsymbol{\beta}}[\mathbf{R}_s],
\end{equation}
which remains affine in $\{\mathbf{W}_{c}^{(k)},\mathbf{W}_{s,n}\}$. 
The auxiliary matrix $\mathbf{Z}$ separates the log-determinant term from the beamforming variables, enabling the SCA-based construction of a tractable surrogate objective. At the $v$-th SCA iteration, let $\mathbf{Z}^{(v)}$ denote the local operating point. Since $-\log|\mathbf{Z}|$ is convex, its first-order Taylor expansion at $\mathbf{Z}^{(v)}$ gives the global affine lower bound \cite{jung2025logdet_SCA}
\begin{equation}
\label{eq:SCA_logdet}
\begin{aligned}
-\log|\mathbf{Z}|
&\ge
-\log|\mathbf{Z}^{(v)}|
-\mathrm{tr}\big((\mathbf{Z}^{(v)})^{-1}(\mathbf{Z}-\mathbf{Z}^{(v)})\big)\\
&=
-\log|\mathbf{Z}^{(v)}|
-\mathrm{tr}\big((\mathbf{Z}^{(v)})^{-1}\mathbf{Z}\big)
+\mathrm{tr}\big(\mathbf{I}_{NM_r}\big),
\end{aligned}
\end{equation}

Substituting \eqref{eq:SCA_logdet} into \eqref{average_KLD} gives the affine lower bound $F_g^{(v)}$ for each angle combination $g\in\mathcal G$. The constant terms are retained because they can differ across angle combinations. To maximize the smallest of these bounds, we introduce a common lower bound $t$. Relaxing the rank constraints then yields the following subproblem

\begin{subequations}\label{p3}
\begin{align}
\text{(P3)}:
& \max_{\{\mathbf{W}_{c}^{(k)},\mathbf{W}_{s,n}\},\,t}\ t \\
\text{s.t.}\quad
& t\leq F_g^{(v)},\quad \forall g\in\mathcal G,\\
& \mathrm{tr}(\mathbf Q_k\mathbf W_c^{(k)})-2\sqrt{N}\delta\|\mathbf W_c^{(k)}\mathbf h_k\|_2 \geq \Gamma_c I_k,\ \forall k,\\
& \sum_{k=1}^K\mathrm{tr}(\mathbf{W}_{c}^{(k)})+\sum_{n=1}^N\mathrm{tr}(\mathbf{W}_{s,n}) \le P^{\max},\\
& \mathbf{W}_{c}^{(k)} \succeq {0},\ \forall k,\quad
\mathbf{W}_{s,n}\succeq0,\ \forall n.
\end{align}
\end{subequations}

For each angle combination, $\mathbf Z$ is affine in the lifted matrices, and (39b) is an affine constraint. The desired-signal lower bound is concave and $I_k$ is convex, so (39c) is also a convex constraint. Therefore, (P3) can be solved using CVX \cite{cvx}.

\begin{algorithm}[t]
\caption{Robust Beamforming for Coherent D-ISAC}
\label{algorithm}
\begin{algorithmic}[1]

\State \algwrap{\textbf{Input:} Feasible initial beam covariances with nonzero sensing power, joint AoA grid, and tolerance $\epsilon$.}
\State \textbf{Output:} ISAC beamforming vectors $\{\mathbf{W}_{c,n}\}$ and $\{\mathbf{w}_{s,n}\}$

\State Set $v=0$ and form $\mathbf Z^{(0)}$ for each angle tuple.

\Repeat

\State \algwrap{Construct (P3) by linearizing $-\log|\mathbf Z|$ for each angle tuple.}

\State Solve (P3) to obtain the beam covariances.

\State \algwrap{Update the linearization point using the current point and the solution of (P3).}
\State Update iteration index $v \leftarrow v+1$

\Until{\algwrap[\algorithmicuntil\ ]{$|F^{(v)}-F^{(v-1)}|\leq\epsilon\max(1,|F^{(v-1)}|)$ or the iteration budget is reached}}

\If{$\mathrm{rank}(\mathbf{W}_{c}^{(k)})>1$ or $\mathrm{rank}(\mathbf{W}_{s,n})>1$}

\State \algwrap{Extract principal eigenvectors and, if needed, reallocate their powers.}
\EndIf
\State \algwrap{Return $\{\mathbf{W}_{c,n}\}$ and $\{\mathbf{w}_{s,n}\}$ if they satisfy \eqref{p1b} and \eqref{p1c}.}
\end{algorithmic}
\end{algorithm}

\textbf{Remark 3: } We obtain feasible initial covariances by maximizing sensing power under the relaxed SINR and power constraints. At each iteration, we update the linearization point by taking a weighted average with the solution of (P3). From the final solution of (P3), we recover the beamformers using principal eigenvector extraction \cite{Sankar2024Rank1}. If needed, we reallocate power along the recovered beam directions subject to (39c) and the power constraints, which are linear in the powers for fixed directions.

\subsection{Convergence Analysis and Computational Complexity}

\subsubsection{Convergence Analysis}

Let $F^{(v)}$ denote the value of the worst-case sensing objective in \eqref{worst_case_aoa} at iteration $v$. The affine bounds are tight at the current point, and their minimum is concave. A convex combination of the current point and the solution of (P3) preserves feasibility and maintains or increases that minimum. Since each affine bound lower-bounds the objective at the corresponding angles, we obtain
\begin{equation}
F^{(v+1)}\geq F^{(v)}.
\end{equation}
Hence, the rank-relaxed objective sequence is monotonically non-decreasing. Moreover, the feasible set is compact due to the total power constraint, and the objective function remains finite over this set. Therefore, the generated sequence converges to a finite value.

\subsubsection{Computational Complexity}

The dominant computational complexity of the proposed method arises from solving the convex subproblem in (P3) at each SCA iteration. The optimization variables include the lifted communication covariance matrices $\{\mathbf{W}_{c}^{(k)}\}$ of dimension $NM_t \times NM_t$, the sensing covariance matrices $\{\mathbf{W}_{s,n}\}$ of dimension $M_t \times M_t$, and a scalar epigraph variable.

Each combination of node angles adds one affine constraint. The grid contains $2^N$ such combinations for each nonzero offset. After representing the norms by semidefinite constraints, let $V$ denote the order of the positive semidefinite variable in a standard-form representation of (P3), and $C$ the number of affine equality constraints. The interior-point method in \cite{Jiang2020Complexity} gives the following time bound for solving one such subproblem to relative accuracy $\varepsilon$, where $\omega$ is the matrix multiplication exponent
\begin{equation}
C_t=\widetilde{\mathcal{O}}\!\left(\sqrt{V}\left(CV^2 + C^\omega + V^\omega\right)\log(1/\varepsilon)\right).
\end{equation}

\section{Numerical Simulation Results}
\label{Numerical Simulation Results}

\begin{table}[t]
    \vspace{-2pt}
    \caption{Simulation parameter setting}
    \label{parameter}
    \centering
    \vspace{-4pt}
    \begin{tabular}{cccc} 
        \\[-1.8ex]\hline 
        \hline \\[-1.8ex] 
        \multicolumn{1}{c}{Parameter} & \multicolumn{1}{c}{Value} & \multicolumn{1}{c}{Parameter} & \multicolumn{1}{c}{Value}\\
        \hline \\[-1.8ex] 
        $M_t=M_r=M$ & $12$   &{$P^{\max}/\sigma_c^2$} &  $30\,$dB \\{$P^{\max}/\sigma_s^2$} &  $30\,$dB    &{$R$} &  $100\,$m\\
        {$\lambda$} & $0.03$\,m & $\gamma$ & $5$\,dB\\
        {$\Delta\theta_{s,n}$} & $2$\,deg   &{$K$} &  $3$ \\
        {$T$} & $100$  &{$\epsilon$} &  $0.01$\\
        {$\delta$} & $0.01$ & $z$ & $30$\,m \\
        \\[-1.8ex]\hline  
        \hline \\[-1.8ex] 
    \end{tabular}
\end{table}

In this section, numerical results are presented to evaluate the performance of the proposed robust beamforming algorithm in the considered coherent D-ISAC system. We first verify the convergence behavior of the proposed algorithm. Subsequently, we investigate the communication--sensing trade-off achieved through joint beamforming optimization across spatially distributed ISAC nodes. Furthermore, we examine the impact of power-constraint configurations and angle-dependent RCS modeling on the sensing performance. Unless otherwise specified, the RCS coefficients follow a Chi-square distribution.

Unless explicitly specified, the system parameters are set according to Table~\ref{parameter}. The target is located above the center of a ground semicircle of radius $R=100$~m. The ISAC nodes are placed along the semicircle, and the UEs are randomly placed on a concentric semicircle of radius $R/2$. The mean stacked downlink channel power and the mean sensing propagation magnitude are normalized to one. The CSI uncertainty radius $\delta$ is applied after downlink channel normalization. The RCS profiles in Fig.~\ref{RCS_distribution} are used as angle-dependent power statistics, and the clutter power per receive antenna is $-10$~dB relative to the sensing noise variance. The optimization problems are implemented using CVX with the MOSEK solver. In all reported simulations, the recovered beamformers satisfied the robust SINR and power constraints, with negligible loss in the worst-case KLD lower bound over the joint AoA grid.

As a benchmark scheme, we consider a \textit{zero-forcing} (ZF)-based beamforming strategy. The sensing beams are formed by projecting the conjugate target steering vectors onto the null spaces of the local UE channels \cite{Liu2018JointTransmission}. The total transmit power is divided using a power allocation ratio $\rho\in[0,1]$, with $\rho P^{\max}$ assigned to communication and $(1-\rho)P^{\max}$ to sensing. A one-dimensional search selects the smallest $\rho$ satisfying the SINR requirements. We further include a \textit{non-robust} counterpart that optimizes the beamformers using the same framework as the proposed method, with the uncertainty models replaced by their nominal counterparts. In particular, it uses the estimated AoAs and unit RCS powers for the sensing objective and the nominal downlink channels for the SINR constraints.

\subsection{Convergence of the Proposed Robust Beamforming Algorithm}
\begin{figure}[t]
    \centering
    \includegraphics[width=1\columnwidth]{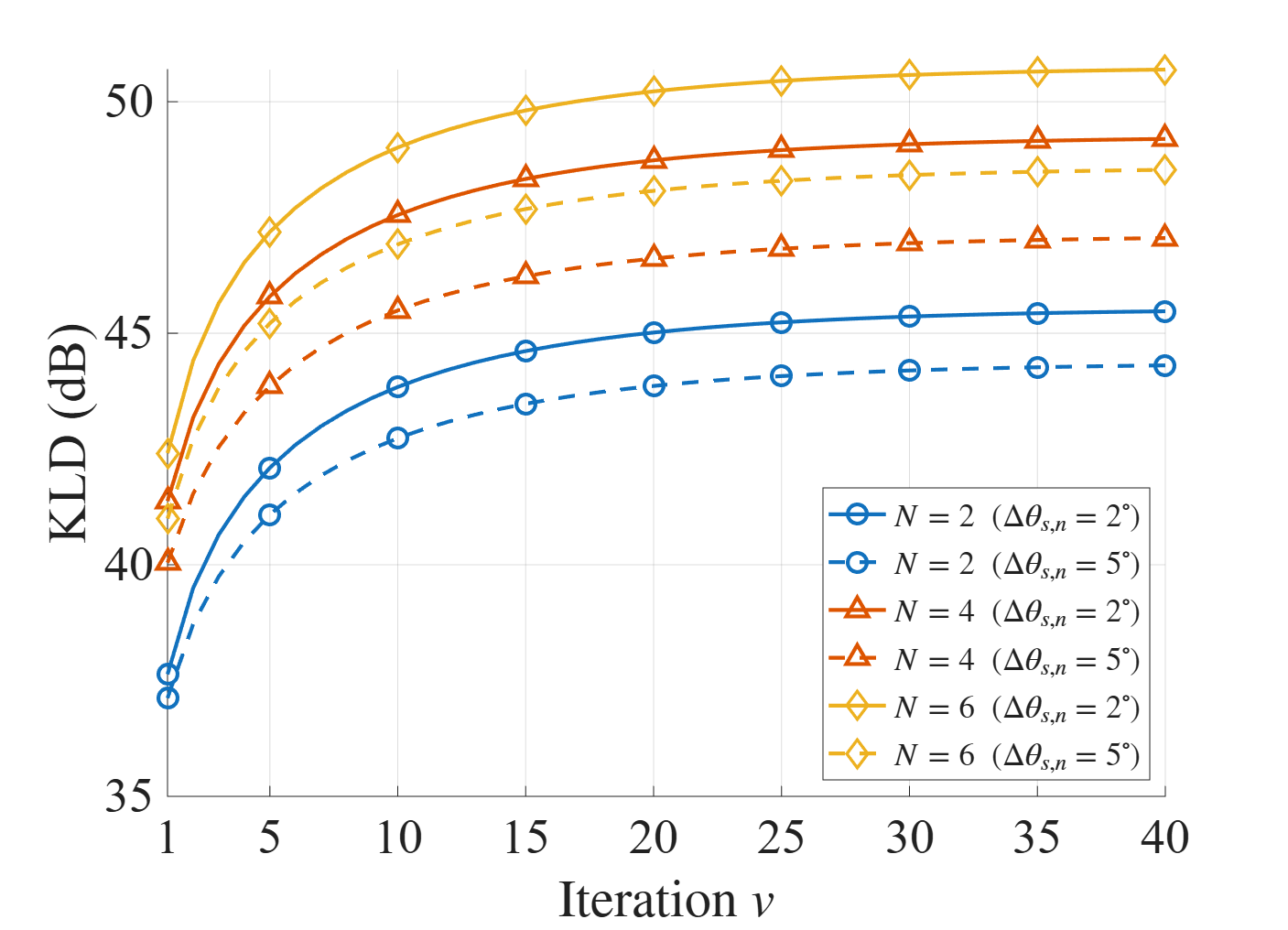}
    \caption{Convergence of the proposed robust beamforming algorithm with $K=3$, $\Gamma_c=20$ dB. }
    \label{Result1_convergence}
\end{figure}

We first investigate the convergence behavior of the proposed robust beamforming algorithm. To examine the impact of system scale and angular uncertainty, we vary the number of distributed ISAC nodes $N \in \{2,4,6\}$ and the AoA uncertainty bound $\Delta\theta_{s,n} \in \{2^\circ, 5^\circ\}$. As shown in Fig.~\ref{Result1_convergence}, the objective increases rapidly in the early iterations and begins to level off after about 30 iterations. For the same AoA uncertainty bound, increasing the number of ISAC nodes improves the KLD. A smaller AoA uncertainty bound also improves the KLD by limiting angular mismatch. In particular, the KLD for $N=4$ and $\Delta\theta_{s,n}=2^\circ$ exceeds that for $N=6$ and $\Delta\theta_{s,n}=5^\circ$, illustrating that tighter AoA uncertainty bounds can provide a larger sensing gain than adding nodes in this setting.

\subsection{Sensing and Communication Performance Trade-off with Robust D-ISAC Beamforming}
\begin{figure}[t]
    \centering
    \includegraphics[width=1\columnwidth]{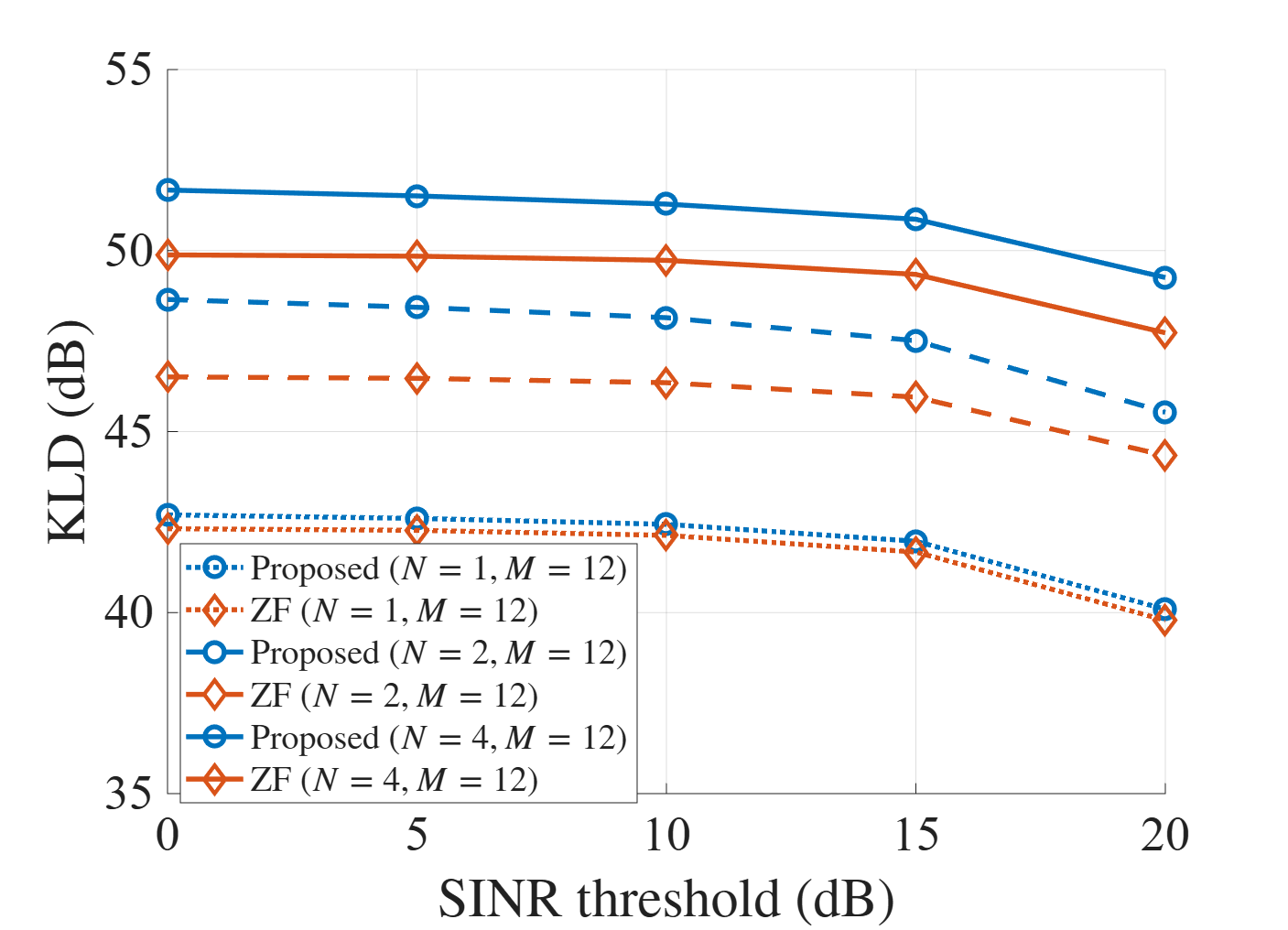}
    \caption{KLD performance of the proposed robust beamforming and the ZF benchmark versus SINR threshold $\Gamma_c$ in a D-ISAC system with $K=3$, and $\Delta\theta_{s,n}=2^{\circ}$.}
    \label{Result2_nodenumber_SINRvsKLD}
\end{figure}

We next examine the SINR--KLD trade-off, as illustrated in Fig.~\ref{Result2_nodenumber_SINRvsKLD}. We compare $N\in\{1,2,4\}$ with $M=12$ antennas per node and the same total transmit power. Both schemes are evaluated using the same AoA grid and robust SINR requirements. As the SINR threshold $\Gamma_c$ increases, less power is available for sensing and the KLD decreases. The results clearly demonstrate the inherent trade-off between the communication performance and the sensing performance for all considered schemes. Compared to the ZF-based strategy that directly allocates beams toward either UEs or the target, the proposed joint communication--sensing beamforming design consistently achieves superior sensing performance across the SINR range. Increasing the number of nodes also improves the KLD for both schemes, as each additional node adds sensing links and antennas.

\begin{figure}[t]
    \centering
    \includegraphics[width=.95\columnwidth]{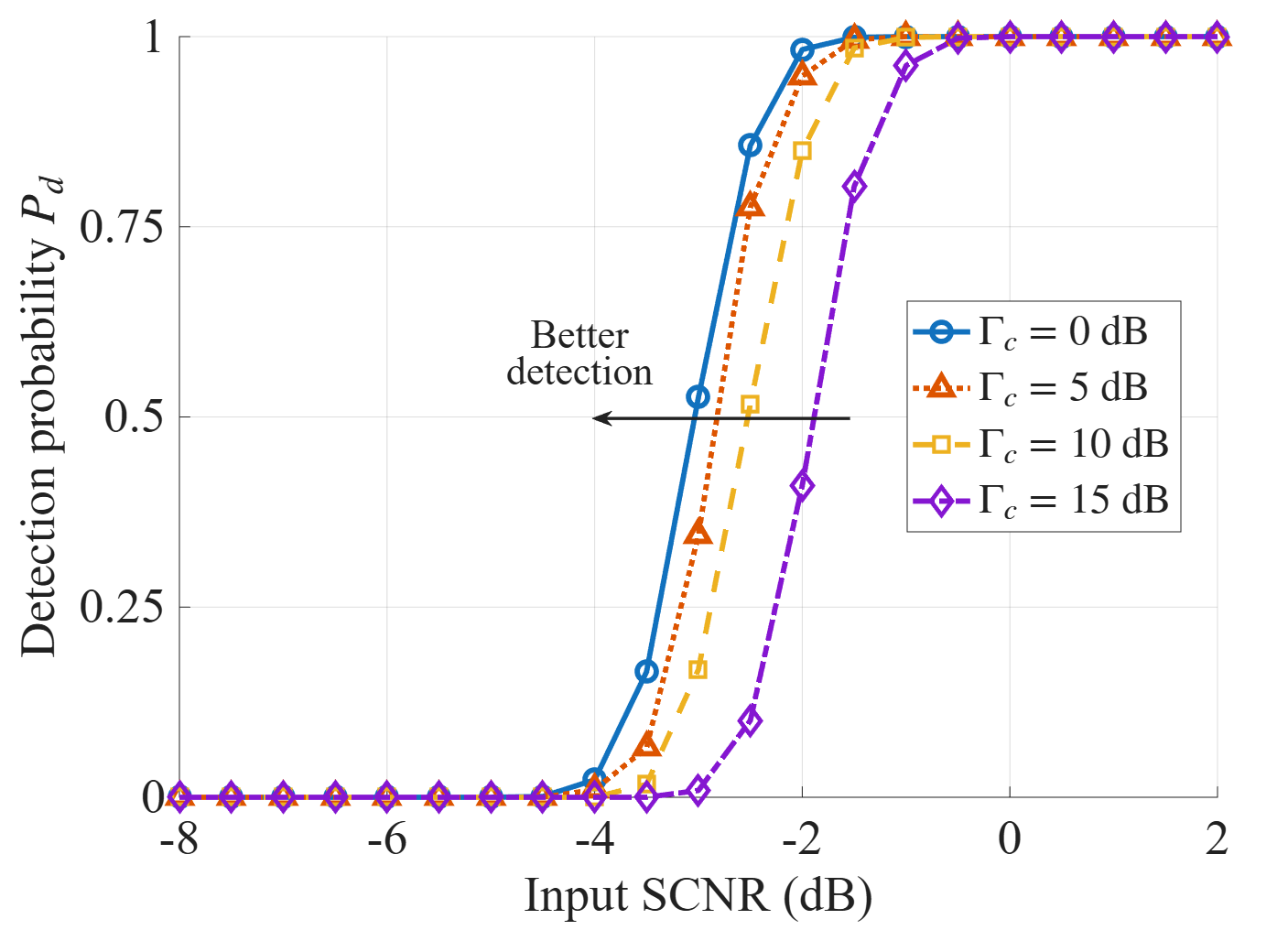}
    \caption{Detection probability versus input SCNR under different SINR thresholds $\Gamma_c$, with $K=3$, $N=2$, $M=12$, and $\Delta\theta_{s,n}=2^{\circ}$.}
    \label{Result3_SINR_SCNRvsPd}
\end{figure}

Furthermore, Fig.~\ref{Result3_SINR_SCNRvsPd} illustrates the detection probability of the proposed beamforming method with $2^\circ$ of target AoA uncertainty. The detection probability is evaluated via Monte Carlo simulations with 1,000 independent realizations of transmitted symbols and noise samples. We obtain the output signal-to-clutter-plus-noise ratio (SCNR) using a whitened matched filter based on the block-diagonal sensing covariance. We subtract the unit clutter-plus-noise power from the average filter output power, clip the result at zero, and apply a detection threshold of 10~dB. The plotted detection probability is the fraction of trials in which the estimated output SCNR exceeds this threshold. All curves use the same noise level at each input SCNR, set using the $\Gamma_c=0$~dB design. As the input SCNR increases, the detection probability rises sharply. Smaller SINR thresholds generally lead to higher detection probability for the same input SCNR, since relaxing the communication QoS constraint allows more power to be allocated to sensing.

\begin{figure}[t]
    \centering
    \includegraphics[width=.95\columnwidth]{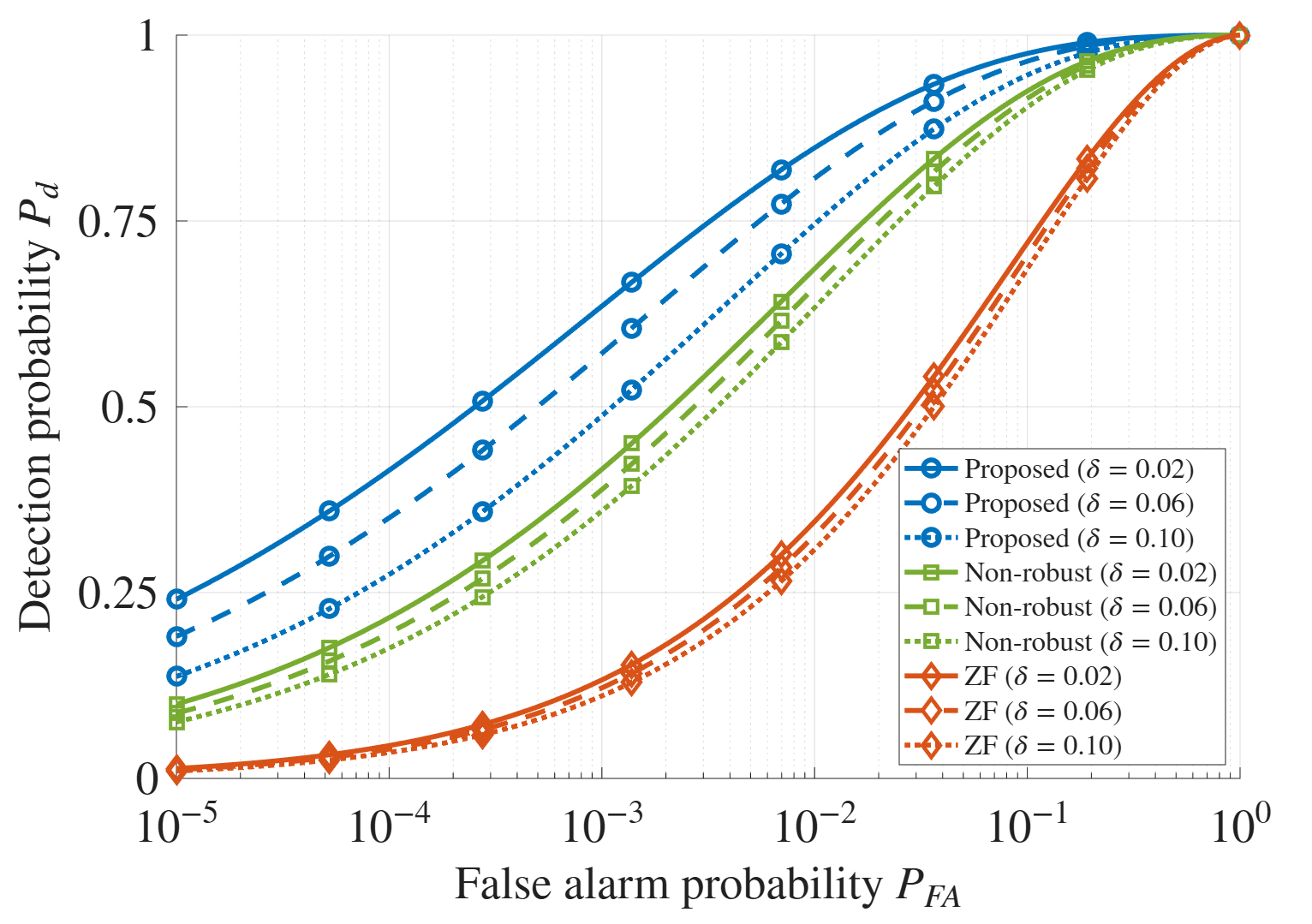}
    \caption{Receiver operating characteristic curves under different CSI uncertainty radii $\delta$, with $K=3$, $N=2$, $M=12$, and $\Delta\theta_{s,n}=2^{\circ}$.}
    \label{Result6_RoC}
\end{figure}

Fig.~\ref{Result6_RoC} presents the worst-case receiver operating characteristic (ROC) curves at $\Gamma_c=0$~dB. For each CSI uncertainty radius $\delta$, we use the largest node phase bound allowed by \eqref{phase_csi_radius} and evaluate each fixed design over the joint AoA grid and the admissible node phases. The whitened matched filter uses the joint receive covariance averaged over the tested AoAs and independent uniform node phases for the proposed scheme, and the nominal covariance for the non-robust and ZF benchmarks. The ROC curves therefore reflect the combined effects of transmit beamforming and receive filtering for each scheme. We verify the minimum filter output power by applying the node phase rotations directly to the echoes. A larger $\delta$ widens the phase-error interval and reduces the worst-case detection probability. At a given false-alarm probability, the non-robust scheme improves on ZF, while the proposed scheme achieves the highest worst-case detection probability.

\subsection{Impact of Power Constraints and Target RCS Modeling}
\label{Impact of Power Constraint Setting and RCS Modeling}

\begin{figure}[t]
    \centering
    \includegraphics[width=.95\columnwidth]{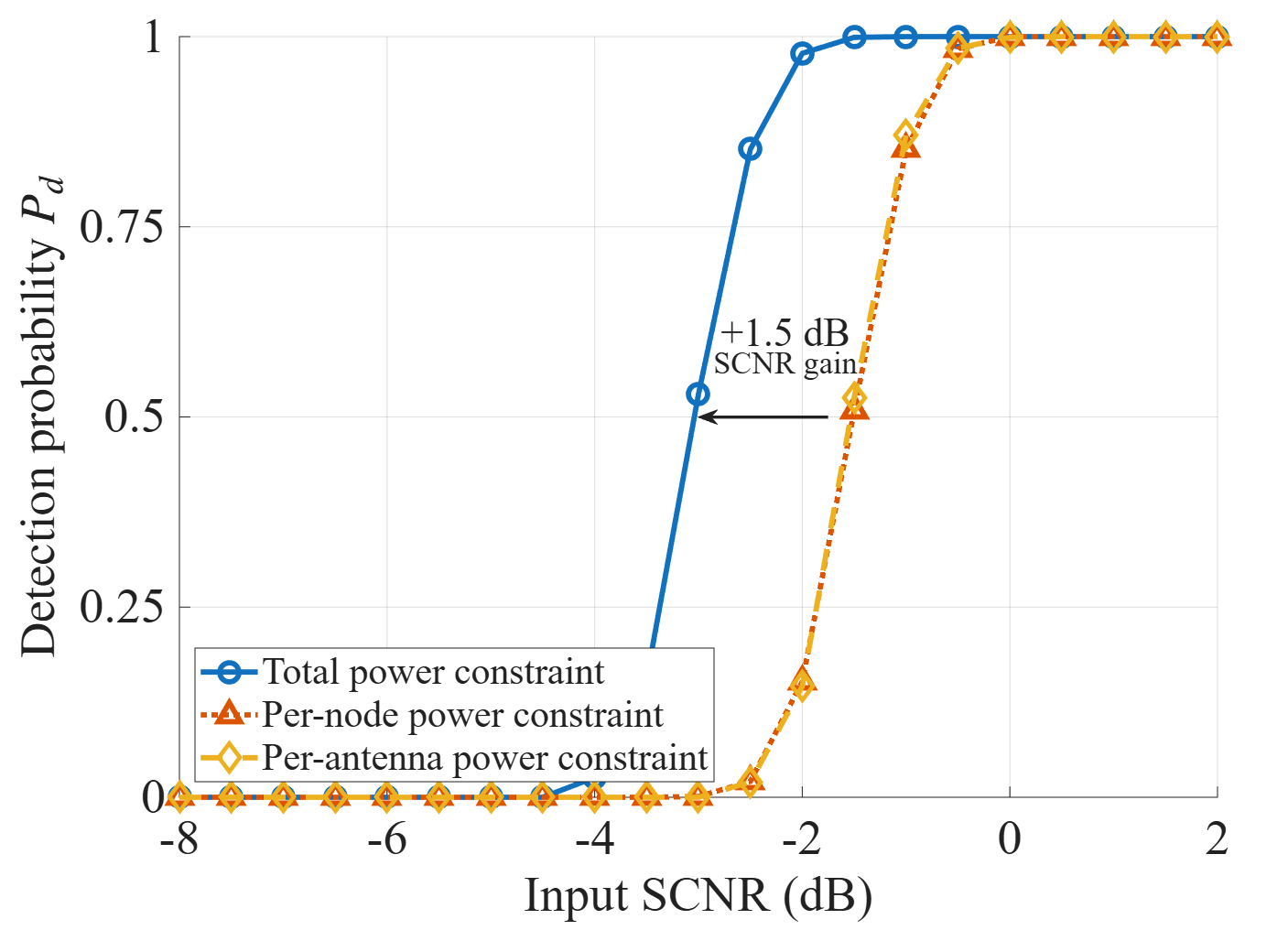}
    \caption{Detection probability versus input SCNR under different power-constraint settings, with $K=3$, $N=2$, $M=12$, $\Gamma_c=0$~dB, and $\Delta\theta_{s,n}=2^{\circ}$.}
    \label{Result4_power_SCNRvsPd.png}
\end{figure}
Fig.~\ref{Result4_power_SCNRvsPd.png} illustrates the detection probability under different power-constraint settings. 
To investigate how the available transmit power budget affects sensing performance, we compare three power allocation strategies: total system power constraint, per-node power constraint, and per-antenna power constraint. As observed in Fig.~\ref{Result4_power_SCNRvsPd.png}, the per-node and per-antenna power constraints result in almost identical detection probability curves. The total system power constraint achieves superior detection performance, corresponding to an SCNR gain of approximately 1.5~dB over the per-node constraint at $P_d=0.5$. This improvement is attributed to the increased flexibility in allocating the transmit power across distributed nodes.

\begin{figure}[t]
    \centering
    \includegraphics[width=.95\columnwidth]{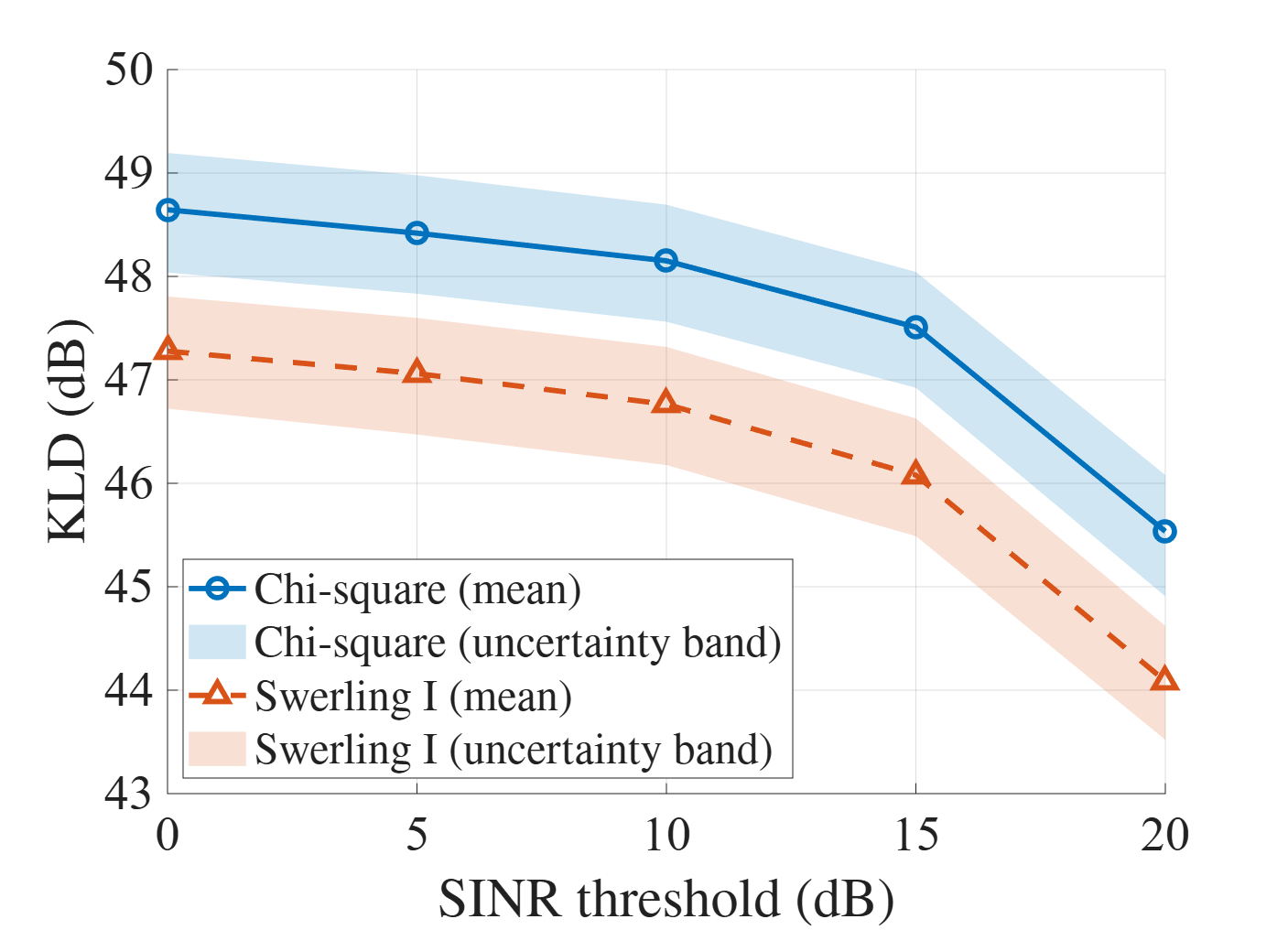}
    \caption{KLD performance versus SINR threshold $\Gamma_c$ under different RCS statistical models, with $K=3$, $N=2$, $M=12$, and $\Delta\theta_{s,n}=2^{\circ}$.}
    \label{Result5_RCS_SINRvsKLD}
\end{figure}
We next examine the impact of statistical RCS modeling on the sensing--communication trade-off. In the simulations, we consider two representative statistical models, namely the Chi-square and the Swerling I model, as illustrated in Fig.~\ref{RCS_distribution}. In Fig.~\ref{Result5_RCS_SINRvsKLD}, the mean KLD decreases for both profiles as the SINR threshold increases. At the evaluated angles, the Chi-square profile yields stronger mean echoes than Swerling I for nearly the same sensing transmit power, which explains its higher mean KLD. With the Chi-square beamformers, receive filter, and threshold held fixed, switching to the Swerling I profile reduces detection probability at the same false-alarm probability. Matching the mean RCS power on each link recovers the original detection performance in this setting.
\section{Conclusion}
\label{Conclusion}

In this paper, we investigated robust beamforming for distributed sensing of the low-observable target and coherent JT-CoMP communication. Under statistical RCS variations, the design maximizes a lower bound on the worst-case expected sensing KLD over a joint AoA grid while satisfying downlink SINR requirements under phase-induced CSI uncertainty. The resulting nonconvex problem was addressed using SDR and SCA. Numerical results show higher worst-case detection probability than the non-robust and ZF benchmarks, with the recovered beamformers satisfying the robust SINR and power constraints. The results also illustrate how distributed deployment, transmit power constraints, and mismatch in the assumed RCS statistics affect sensing performance. Future work will consider multi-target detection and OFDM-based D-ISAC systems.

\appendix
\subsection{Proof of Lemma 1}
\label{proof_lemma1}

For the desired-signal term in the numerator, let $\mathbf W=\mathbf W_c^{(k)}\succeq0$ and write the effective stacked channel as $\tilde{\mathbf h}_k=\mathbf h_k+\mathbf e_k$. Expanding the desired-signal power yields
\begin{equation}
\begin{aligned}
\tilde{\mathbf h}_k^H\mathbf W\tilde{\mathbf h}_k&=\mathbf h_k^H\mathbf W\mathbf h_k+2\Re(\mathbf e_k^H\mathbf W\mathbf h_k)\\
&\quad+\mathbf e_k^H\mathbf W\mathbf e_k\\
&\geq \mathbf h_k^H\mathbf W\mathbf h_k-2\sqrt N\delta\|\mathbf W\mathbf h_k\|_2,
\end{aligned}
\end{equation}
where
$\tilde{\mathbf h}_k=\mathbf h_k+\mathbf e_k$,
$\mathbf e_k\triangleq[\mathbf e_{1,k}^T,\ldots,\mathbf e_{N,k}^T]^T$,
$\mathbf h_k\triangleq[\mathbf h_{1,k}^T,\ldots,\mathbf h_{N,k}^T]^T$,
$\mathbf w_c^{(k)}\triangleq[\mathbf w_{c,1,k}^T,\ldots,\mathbf w_{c,N,k}^T]^T$,
$\mathbf W_c^{(k)}\triangleq\mathbf w_c^{(k)}(\mathbf w_c^{(k)})^H$, and
$\mathbf Q_k\triangleq\mathbf h_k\mathbf h_k^H$.
Since $\|\mathbf e_k\|_2^2=\sum_{n=1}^N\|\mathbf e_{n,k}\|_2^2\leq N\delta^2$, the last inequality follows from the Cauchy--Schwarz inequality and $\mathbf e_k^H\mathbf W\mathbf e_k\geq0$. Using $\mathbf h_k^H\mathbf W\mathbf h_k=\mathrm{tr}(\mathbf Q_k\mathbf W)$ and minimizing the left-hand side over the admissible errors yields \eqref{signal_min}.

For the inter-user interference term in the denominator, we first note that $
\big\lvert \sum_{n=1}^N \tilde{\mathbf{h}}_{n,k}^H\mathbf{w}_{c,n,j} \big\rvert^2
=
\mathrm{tr}\!\left(
\tilde{\mathbf h}_k\tilde{\mathbf h}_k^H
\mathbf W_{c}^{(j)}
\right)$. 
By expanding this term with $\tilde{\mathbf h}_k=\mathbf h_k+\mathbf e_k$ and applying the Cauchy--Schwarz inequality, we obtain
\begin{equation}
\begin{aligned}
&\max_{\lVert \mathbf{e}_k\rVert ^2\leq N\delta^2 }
\mathrm{tr}\!\left(
\tilde{\mathbf h}_k\tilde{\mathbf h}_k^H
\mathbf W_{c}^{(j)}
\right)
\\&\leq
\mathrm{tr}\!\left(\mathbf Q_k \mathbf W_{c}^{(j)}\right)
+ N\delta^2\,\mathrm{tr}(\mathbf W_{c}^{(j)})
\\&+ \sqrt{N}\delta \big\lVert\mathbf{W}_{c}^{(j)}\mathbf{h}_k\big\rVert
+ \sqrt{N}\delta \big\lVert\mathbf{h}_k^H\mathbf{W}_{c}^{(j)}\big\rVert.
\end{aligned}
\end{equation}

Similarly, the sensing interference can be bounded using the same approach. To this end, we can obtain
\begin{equation}
\begin{aligned}
&\max_{\{\mathbf e_{n,k}\}:\lVert\mathbf{e}_{n,k}\rVert ^2 \leq \delta^2}
\sum_{n=1}^N
\big\lvert  \tilde{\mathbf{h}}_{n,k}^H\mathbf{w}_{s,n} \big\rvert^2
\\&\leq
\sum_{n=1}^N
\max_{\lVert\mathbf{e}_{n,k}\rVert ^2 \leq \delta^2}
\mathrm{tr}\!\left(
\tilde{\mathbf h}_{n,k}\tilde{\mathbf h}_{n,k}^H
\mathbf W_{s,n}
\right) \\
&\leq
\sum_{n=1}^N
\Big(
\mathrm{tr}(\mathbf Q_{n,k}\mathbf W_{s,n})
+ \delta^2\,\mathrm{tr}(\mathbf W_{s,n}) 
+ \\&\delta\lVert\mathbf{W}_{s,n}\mathbf{h}_{n,k}\rVert
+ \delta\lVert \mathbf{h}_{n,k}^H\mathbf{W}_{s,n}\rVert
\Big).
\end{aligned}
\end{equation}
By substituting the above communication- and sensing-interference bounds into the denominator of \eqref{min_max_SINR} and adding the noise power, we obtain the upper bound in \eqref{interference_max}.

\bibliographystyle{IEEEtran}
\bibliography{ref}

\end{document}